\documentclass[journal,draftcls,onecolumn,12pt,twoside]{IEEEtran}

\usepackage{times}
 \usepackage{amsmath}
 \usepackage{amssymb}
 \usepackage{amsthm}
 \usepackage{color}
  \usepackage{algorithm}
 \usepackage[noend]{algorithmic}
      \usepackage{graphicx}
\usepackage{subfig}
\usepackage{multirow}
 \usepackage{url}
\usepackage{cite}
\usepackage{bm}

\newtheorem{theorem}{Theorem}

\newtheorem{lemma}[theorem]{Lemma}

\newtheorem{defn}{Definition}
\newtheorem{example}{Example}

\normalsize

\ifCLASSINFOpdf
\else
\fi

\begin{document}
%
\title{A New Design of  Binary MDS Array Codes with Asymptotically Weak-Optimal Repair}

\author{Hanxu Hou,~\IEEEmembership{Member,~IEEE}, Yunghsiang S. Han,~\IEEEmembership{Fellow,~IEEE}, Patrick P. C. Lee,~\IEEEmembership{Senior Member,~IEEE}, Yuchong Hu,~\IEEEmembership{Member,~IEEE}, and Hui Li,~\IEEEmembership{Member,~IEEE}
\thanks{This paper was presented in part in \cite{hanxu2017triple} at the IEEE International Symposium on Information Theory, Aachen, Germany, June 2017.
H. Hou is with the School of Electrical Engineering \& Intelligentization, Dongguan University of Technology and the Shenzhen Graduate School, Peking University~(E-mail: houhanxu@163.com).
Y. S. Han is with the School of Electrical Engineering \& Intelligentization, Dongguan University of Technology~(E-mail: yunghsiangh@gmail.com).
P. P. C. Lee is with Department of Computer Science and Engineering, The Chinese University of Hong Kong~(E-mail: pclee@cse.cuhk.edu.hk).
Y. Hu is with the School of Computer Science and Technology, Huazhong University of
Science and Technology (E-mail: yuchonghu@hust.edu.cn).
H. Li is with the Shenzhen Graduate School, Peking University~(E-mail: lih64@pkusz.edu.cn).
This work was partially supported by the National Natural Science Foundation of China (No. 61701115, 61671007, 61671001) and National Keystone R\&D Program of China (No. 2017YFB0803204, 2016YFB0800101). }
}

\markboth{IEEE Transactions on Information Theory}%
{Submitted paper}

\maketitle
\vspace{-1.8cm}
\begin{abstract}Binary maximum distance separable (MDS) array codes are a special class of erasure codes for distributed storage that not only provide fault tolerance with minimum storage redundancy but also achieve low computational complexity. They are constructed by encoding $k$ information columns into $r$ parity columns, in which each element in a column is a bit, such that any $k$ out of the $k+r$ columns suffice to recover all information bits. In addition to providing fault tolerance, it is critical to improve repair performance in practical applications. Specifically, if a single column fails, our goal is to minimize the repair bandwidth by downloading the least amount of bits from $d$ healthy columns, where $k\leq d\leq k+r-1$. If one column of an MDS code is failed, it is known that we need to download at least $1/(d-k+1)$ fraction of the data stored in each of $d$ healthy columns. If this lower bound is achieved for the repair of the failure column from accessing arbitrary $d$ healthy columns, we say that the MDS code has optimal repair. However, if such lower bound is only achieved by $d$ specific healthy columns, then we say the MDS code has weak-optimal repair.
In this paper, we propose two explicit constructions of binary MDS array codes with more parity columns (i.e., $r\geq 3$) that achieve asymptotically weak-optimal repair, where $k+1\leq d\leq k+\lfloor(r-1)/2\rfloor$, and ``asymptotic" means that the repair bandwidth achieves the minimum value asymptotically in $d$. Codes in the first construction have odd number of parity columns and asymptotically weak-optimal repair for any one information failure, while codes in the second construction have even number of parity columns and asymptotically weak-optimal repair for any one column failure. 
\end{abstract}

\IEEEpeerreviewmaketitle

\section{Introduction}
Modern distributed storage systems deploy erasure codes to maintain data
availability against failures of storage nodes.  Binary maximum distance
separable (MDS) array codes are a special class of erasure codes that achieve
fault tolerance with 
minimum storage redundancy and low computational complexity.  Specifically, a
binary array code is composed of $k+r$ columns with $L$ bits in each column.
Among the $k+r$ columns, $k$ \emph{information columns} store information bits
and $r$ \emph{parity columns} store parity bits. The $L$ bits in each
column are stored in the same storage node.
We refer to a disk as a column or a storage node interchangeably, and an
entry in the array as a bit.  When a node fails, the corresponding column of
the array code is considered as an \emph{erasure}.
A code is said to be MDS if any $k$ out of the $k+r$ columns suffice to
reconstruct all $k$ information columns. Hence, an MDS code can tolerate any $r$ failed
columns.  
Examples of binary MDS array codes include double-fault tolerance codes (i.e., $r=2$) such as X-code~\cite{xu1999x}, RDP codes~\cite{corbett2004row} and EVENODD codes~\cite{blaum1995evenodd}, and  triple-fault tolerance codes (i.e., $r=3$) such as STAR codes~\cite{huang2008star}, generalized RDP codes~\cite{blaum2006family}, and TIP codes \cite{zhang2015tip}.


When a node fails in a distributed storage system, one should repair the failed
node by downloading bits from $d$ healthy (helper) nodes, 
where $k\le d\le k+r-1$.  Minimizing the {\em repair bandwidth}, defined as
the amount of bits downloaded in the repair process, is critical to speed up
the repair operation and minimize the window of vulnerability, especially in
distributed storage in which network transfer is the bottleneck.
The repair problem was first formulated by Dimakis {\em et al.}
\cite{dimakis2010network} based on the concept of information flow graph.
It is shown in \cite{dimakis2010network} that the minimum repair bandwidth
subject to the minimum storage redundancy, also known as the
{\em minimum storage regenerating (MSR)} point, is given by:
\begin{equation}
\frac{dL}{d-k+1}.
\label{optimal_repair}
\end{equation}
If the lower bound in \eqref{optimal_repair} is achieved for the repair of the failure node from accessing arbitrary $d$ healthy nodes, we say that the MDS code has {\em optimal repair}. When such lower bound is only achieved by accessing $d$ {\em specific} healthy nodes instead of arbitrary $d$ healthy nodes, we then say the MDS code has  {\em weak-optimal repair}. Note that the above repair is exact repair, i.e., the content stored in the failed node is exactly reproduced in the new node. Another version of repair is functional repair, i.e., the recovered content and the failed content may be different. In functional repair, it is interesting to note that there may exist a helper node selection that can improve the lower bound in \eqref{optimal_repair} for some specific parameters $n,k,d$ \cite{Ahmad2017When} when the $d$ helper nodes are specified.
Although there are many constructions of MSR codes 
\cite{dimakis2010network,tamo2013zigzag,Ye2016Explicit} over a
sufficiently large finite field, the constructions of binary MDS array codes that
achieve the minimum repair bandwidth is not many \cite{Sasidharan2017An}.

A conventional approach for repairing a failed node is to download all the bits from any $k$ healthy columns to regenerate the bits in the failure column. As a result, the total amount of bits downloaded to repair a failure column is $k$ times of the failure bits. There have been many studies on reducing the repair bandwidth for a single failed
column in binary MDS array codes.  Some approaches minimize disk reads for RDP
codes \cite{xiang2010optimal}, EVENODD codes~\cite{wang2010rebuilding} and X-code~\cite{xu2014single} with $d=k+1$; however, their repair
bandwidth is sub-optimal and 50\% larger than the lower bpund in
\eqref{optimal_repair} when $d=k+1$.  MDR codes \cite{Wang2013MDR,Wang2016Two} and
ButterFly codes \cite{gad2013repair,pamies2016opening} are binary MDS array codes that achieve
optimal repair; however, they only provide double-fault tolerance (i.e.,
$r\!=\!2$).  

In this paper, we propose two explicit constructions of binary MDS array codes with more parity columns (i.e., $r\geq 3$) that achieve asymptotically weak-optimal repair,  where``asymptotic"  means that the repair bandwidth achieves the minimum value asymptotically in $d$.  The main contributions of this paper are as follows.
\begin{enumerate}
\item We exploit a quotient ring with cyclic structure and propose a new approach of designing binary MDS array codes with $r\geq 3$ parity columns based on the quotient ring. 
\item Two explicit constructions of binary MDS array codes with $r\geq 3$ parity columns are presented based on the proposed  approach. Codes in the first construction have odd number of parity columns that is formed by designing an encoding matrix and codes in the second construction have even number of parity columns that is formed by designing a check matrix.
\item Our constructions minimize the repair bandwidth by exploiting the proposed quotient ring and choosing the well-designed encoding matrix (parity matrix), such that the bits accessed in a repair operation intersect as much as possible.
\item  We show that the first construction of the proposed binary MDS array codes has asymptotically weak-optimal repair with $d=k+(r-1)/2$ for any single information column failure  when $d$ is sufficiently large. We also show that the second construction of the proposed binary MDS array codes has asymptotically weak-optimal repair with $d=k+r/2-1$ for any column failure when $d$ is large enough. 
\end{enumerate}
 
  
The repair bandwidth of most existing binary MDS array codes \cite{xu1999x,corbett2004row,blaum1995evenodd,blaum2006family,huang2008star} is sub-optimal. Some constructions \cite{Wang2013MDR,Wang2016Two,gad2013repair,pamies2016opening} of binary MDS array codes with optimal repair bandwidth only focus on double-fault tolerance (i.e., $r\!=\!2$). 
The key differences between the proposed codes and the existing binary MDS array codes are as follows. First, in contrast to existing constructions such as \cite{xu1999x,corbett2004row,blaum2006family,huang2008star}, the parity bits in parity columns (except the first parity column) of the proposed codes are generated by summing the bits that correspond to a specific polygonal lines in the array. Second, the row number of the array in the proposed codes is exponential in $k$. The two properties are essential for reducing the repair bandwidth. The difference between the double-fault tolerance optimal repair constructions \cite{Wang2013MDR,Wang2016Two,gad2013repair,pamies2016opening} and the proposed constructions is that a quotient ring with cyclic structure is employed in the proposed construction, while is not in \cite{Wang2013MDR,Wang2016Two,gad2013repair,pamies2016opening}. By exploiting the quotient ring, we can choose the well-designed encoding matrix (parity matrix) and achieve the weak-optimal repair bandwidth asymptotically with larger fault tolerance.

Unlike the other binary MDS array codes \cite{Wang2013MDR,Wang2016Two,gad2013repair,pamies2016opening} with optimal repair bandwidth, the proposed  codes with the first construction  have asymptotically weak-optimal repair bandwidth for recovering one information column failure. Note that the information failure column is recovered by accessing $d$ specific columns instead of any arbitrary set of $d$ healthy columns. Similarly, the proposed codes with the second construction have asymptotically weak-optimal repair bandwidth for recovering any one column failure (including both information column and parity column), and the failure column is also recovered by accessing $d$ specific healthy columns. 
Note that there exist binary MDS codes \cite{Sasidharan2017An} that are weak-optimal for any $(n,k,d)$; however, the techniques of the constructions of the MDS codes in \cite{Sasidharan2017An} and this paper are different. The construction in \cite{Sasidharan2017An} employed the coupling transformation over binary-vector MDS codes such that they have weak-optimal repair property. However, the proposed binary codes are constructed directly over a binary quotient ring. In this quotient ring, we select  the values of parameter $p$ to make the code to be MDS and determine  the values of parameter $\tau$ to make the code having asymptotic weak-optimal repair property.

Previous studies~\cite{shumregenerating,Hou2016BASIC} also exploit similar techniques (applying quotient rings) to reduce computational complexity of regenerating codes. In this work, we show that when $\tau$ (a parameter that will be found later) is large enough and satisfies some conditions, we can find some constructions of binary MDS array codes that can achieve the weak-optimal repair asymptotically. The ring in \cite{shumregenerating,Hou2016BASIC} can be viewed as a special case of the proposed ring with $\tau=1$.  Moreover, the main results between \cite{shumregenerating,Hou2016BASIC}  and this paper are different. It was shown in \cite{shumregenerating,Hou2016BASIC}  that the fundamental tradeoff curve between the storage and repair bandwidth of functional-repair regenerating codes is also achievable under a quotient ring, and the existing product-matrix construction of regenerating codes still works  under the quotient ring with less computational complexity. While in this paper, we use a more general ring to construct a new class of binary MDS array codes with asymptotically weak-optimal repair bandwidth by choosing the well-designed generator matrix or parity matrix. Even though the proposed binary MDS array codes and constructions of high data rate MSR codes~\cite {Li2015A,tamo2013zigzag,Ye2016Explicit,ye2017explicit,ye2017explicit1,goparaju2017minimum,tian2017generic,rashmi2017piggybacking} are all based on constructing generator matrices or parity matrices, the proposed codes are constructed over binary field and the  encoding matrices or parity matrices are designed  on the ring with a cyclic structure.


This paper is organized as follows. An approach of designing binary array codes with general number of parity columns is given in Section~\ref{sec:framework}. In Section~\ref{sec:construction},  two constructions of binary MDS array codes with more than three parity columns are presented. In Section~\ref{sec:MDS},  a sufficient condition of the MDS property condition for two constructions is presented. We also illustrate the  MDS property  of $r=3$ for the first construction and $r=4$ for the second construction in detail.
Two efficient repair algorithms are proposed in Section~\ref{sec:repair} for the proposed two constructions to recover single information column erasure and single column erasure respectively.  It is also shown that both constructions can achieve the weak-optimal repair bandwidth in \eqref{optimal_repair} asymptotically. 
We make conclusion and remark in Section~\ref{sec:discussions}.

\section{New Design of Binary MDS Array Codes}
\label{sec:framework}
Consider the  binary MDS array code with $k\geq 2$ information columns and $r\geq 3$ parity columns. Each column of this code stores $(p-1)\tau$ bits, where $p$ is a prime number such that 2 is a primitive element in the field $\mathbb{Z}_p$ and $\tau$ will be specified later.
Assume that a file of size $k(p-1)\tau$ denoted by information bits $s_{0,i},s_{1,i},\ldots, s_{(p-1)\tau-1,i}\in \mathbb{F}_2^{(p-1)\tau}$ for $i=1,2,\ldots,k$,  which are used to generate $r(p-1)\tau$ parity bits $s_{0,j},s_{1,j},\ldots, s_{(p-1)\tau-1,j}\in \mathbb{F}_2^{(p-1)\tau}$  for $j=k+1,k+2,\ldots,k+r$. The information bits $s_{0,i},s_{1,i},\ldots, s_{(p-1)\tau-1,i}$ are stored in information column $i$ (column $i$) for $i=1,2,\ldots,k$, and the $(p-1)\tau$ parity bits $s_{0,j},s_{1,j},\ldots, s_{(p-1)\tau-1,j}$ are stored in parity column $j-k$ (column $j$) for $j=k+1,k+2,\ldots,k+r$.

For $i=1,2,\ldots,k$ and $\mu=0,1,\ldots,\tau-1$, we define the \emph{extra bit} $s_{(p-1)\tau+\mu,i}$ associated with bits $s_{\mu,i},s_{\tau+\mu,i},\ldots,s_{(p-2)\tau+\mu,i}$ as
\begin{equation}
s_{(p-1)\tau+\mu,i} \triangleq  \sum_{\ell=0}^{p-2} s_{\ell\tau+\mu,i}.
\label{eq:check1}
\end{equation}
For example,
when $p=3$, $k=4$ and $\tau=4$, the extra bit of
$s_{\mu,i},s_{4+\mu,i}$ is
$$s_{8+\mu,i}=s_{\mu,i}+s_{4+\mu,i}.$$ 
For $j=k+1,k+2,\ldots,k+r$,  $\tau$ extra bits $s_{(p-1)\tau,j},s_{(p-1)\tau+1,j},\ldots,s_{p\tau-1,j}$ for parity column $j-k$ are added during the encoding procedure. It will be clear later that the extra bit $s_{(p-1)\tau+\mu,j}$ of parity column $j-k$ also satisfies~\eqref{eq:check1} for $j=k+1,k+2,\ldots,k+r$ and $\mu=0,1,\ldots,\tau-1$, if we replace $i$ with  $j$ in~\eqref{eq:check1}.

For $\ell=1,2,\ldots,k+r$, we represent $(p-1)\tau$ bits $s_{0,\ell},s_{1,\ell},\ldots,s_{(p-1)\tau-1,\ell}$ in column $\ell$, together with the $\tau$ extra bits $s_{(p-1)\tau,\ell},s_{(p-1)\tau+1,\ell},\ldots,s_{p\tau-1,\ell}$, by a polynomial $s_{\ell}(x)$ of degree at most $p\tau-1$  over the ring $\mathbb{F}_2[x]$, i.e.,
\begin{equation}
s_\ell(x)\triangleq s_{0,\ell}+s_{1,\ell}x+s_{2,\ell}x^2+\cdots+s_{p\tau-1,\ell}x^{p\tau-1}.
\label{eq:poly}
\end{equation}
The polynomial $s_{i}(x)$, which corresponds to information column $i$ for $i=1,2,\ldots,k$, is called a \emph{data polynomial}; the polynomial $s_{j}(x)$, which corresponds to parity column $j-k$ for $j=k+1,k+2,\ldots,k+r$, is called a \emph{coded polynomial}.
The $k$ data polynomials and $r$ coded polynomials can be arranged as a row vector
\begin{equation}
[s_1(x),s_2(x),\cdots,s_{k+r}(x)].
\label{eq:data_clm}
\end{equation}
\eqref{eq:data_clm} can be obtained by taking the product
\begin{equation}
[s_1(x),\cdots,s_{k+r}(x)]= [s_1(x),\cdots,s_{k}(x)] \cdot \mathbf{G}_{k\times (k+r)}\label{eq:encoding}
\end{equation}
with operations performed in $\mathbb{F}_2[x]/(1+x^{p\tau})$. The $k\times (k+r)$ \emph{generator matrix} $\mathbf{G}_{k\times (k+r)}$ is composed of the $k\times k$ identity matrix $\mathbf{I}_{k\times k}$ and a $k\times r$ \emph{encoding matrix} $\mathbf{P}_{k\times r}$ as
\begin{equation}
\label{generator-matrix}
\mathbf{G}_{k\times (k+r)}=
\begin{bmatrix}
\mathbf{I}_{k\times k} & \mathbf{P}_{k\times r}
\end{bmatrix}.
\end{equation}
The proposed code can also be described equivalently by a $r\times (k+r)$ \emph{check matrix} $\mathbf{H}_{r\times (k+r)}$. Given the row vector in \eqref{eq:data_clm}, we have
\begin{equation}
[s_1(x),s_2(x),\cdots,s_{k+r}(x)] \cdot \mathbf{H}_{r\times (k+r)}^T={\bm 0}.
\label{eq:check}
\end{equation}

Let $R_{p\tau}$ denote the quotient ring $\mathbb{F}_2[x]/(1+x^{p\tau})$. An element $a(x)$ in $R_{p\tau}$ can be represented by a polynomial of the form $a(x)=a_0 + a_1x + \cdots +a_{p\tau-1}x^{p\tau-1}$ with coefficients from the binary field $\mathbb{F}_2$. Addition is the usual term-wise addition, and multiplication is performed with modulo $1+x^{p\tau}$. In $R_{p\tau}$, multiplication by $x$ can be interpreted as a \emph{cyclic shift}. This is crucial to reduce repair bandwidth for one column failure. Note that  we do not need to store the extra bits on disk, they are presented only for notational convenience.

Consider the sub-ring $C_{p\tau}$ of $R_{p\tau}$ which consists of polynomials in $R_{p\tau}$ with $1+x^{\tau}$ being a factor,
\begin{equation}
C_{p\tau} \triangleq \{ a(x) (1+x^{\tau}) \bmod (1+x^{p\tau})|\, a(x) \in R_{p\tau} \}.
\label{eq:C}
\end{equation}

In fact,  $C_{p\tau}$ is an ideal, because
\[
 \forall c(x)\in R_{p\tau}, \forall s(x)\in C_{p\tau},\ c(x)s(x)\in C_{p\tau}.
\]
One can verify that the product of 
$$h(x) \triangleq 1+x^{\tau}+\cdots + x^{(p-1)\tau}$$
and any polynomial in $C_{p\tau}$ is zero. The polynomial $h(x)$ is called the {\em check polynomial} of $C_{p\tau}$. 
The multiplication identity of $C_{p\tau}$ is
\begin{equation}
\begin{array}{ll}
 e(x)&\triangleq 1+h(x) = x^{\tau} + x^{2\tau} + \cdots + x^{(p-1)\tau}\\
 &=(1+x^{\tau})(x^{\tau}+x^{3\tau}+\cdots+x^{(p-2)\tau}),
 \end{array} 
 \label{eq:identity-1}
\end{equation}
as $\forall b(x)\in C_{p\tau}$,
\begin{equation}
e(x)b(x)=(1+h(x))b(x)=b(x)\bmod (1+x^{p\tau}).
\label{eq:identity}
\end{equation}

\begin{theorem}
\label{even-parity}
The coefficients of polynomial $s_{i}(x)$ satisfy \eqref{eq:check1} if and only if $s_{i}(x)$ is in $C_{p\tau}$. 
\end{theorem}
\begin{proof}
Suppose that the coefficients of $s_{i}(x)$ satisfy \eqref{eq:check1}. By reformulating $s_{i}(x)$, we have 
\begin{align*}
\tiny
&s_i(x)=s_{0,i}+s_{1,i}x+\cdots+s_{(p-1)\tau-1,i}x^{(p-1)\tau-1}\\
&+x^{(p-1)\tau}\sum_{\ell=0}^{p-2}s_{\ell\tau,i}+\cdots+x^{p\tau-1}\sum_{\ell=0}^{p-2}s_{\ell\tau+\tau-1,i} \\
=&s_{0,i}+s_{\tau,i}x^{\tau}+\cdots+s_{(p-2)\tau,i}x^{(p-2)\tau}\\
&+x^{(p-1)\tau}\sum_{\ell=0}^{p-2}s_{\ell\tau,i}+\\
&s_{1,i}x+s_{\tau+1,i}x^{\tau+1}+\cdots+s_{(p-2)\tau+1,i}x^{(p-2)\tau+1}\\
&+x^{(p-1)\tau+1}\sum_{\ell=0}^{p-2}s_{\ell\tau+1,i}+ \cdots +\\
&s_{\tau-1,i}x^{\tau-1}+s_{2\tau-1,i}x^{2\tau-1}+\cdots+\\
&s_{(p-1)\tau-1,i}x^{(p-1)\tau-1}+x^{p\tau-1}\sum_{\ell=0}^{p-2}s_{(\ell+1)\tau-1,i}\\
=&s_{0,i}(1+x^{(p-1)\tau})+\cdots+s_{(p-2)\tau,i}x^{(p-2)\tau}(1+x^{\tau})+\\
&s_{1,i}x(1+x^{(p-1)\tau})+\cdots+s_{(p-2)\tau+1,i}x^{(p-2)\tau+1}(1+x^{\tau})\\
&+\cdots +s_{\tau-1,i}(x^{\tau-1}+x^{p\tau-1})\\
&+\cdots+s_{(p-1)\tau-1,i}(x^{(p-1)\tau-1}+x^{p\tau-1}).
\end{align*}
This is reduced to show that $x^{i\tau+j}+x^{(p-1)\tau+j}$ is a multiple of $1+x^{\tau}$ for $i=0,1,\ldots,p-2$ and $j=0,1,\ldots,\tau-1$. It is true from the fact
\begin{align*}
&x^{i\tau+j}+x^{(p-1)\tau+j}=x^{i\tau+j}(1+x^{(p-i-1)\tau})\\
=&x^{i\tau+j}(1+x^{\tau})(1+x^{\tau}+\cdots+x^{(p-i-2)\tau}).
\end{align*}
This verifies that the polynomial $s_{i}(x)$ is in $C_{p\tau}$.

Conversely, suppose that $s_{i}(x)=\sum_{\ell=0}^{p\tau-1}s_{\ell,i}x^\ell$ is in $C_{p\tau}$. By \eqref{eq:C}, $s_{i}(x)$ can be written as 
\begin{align*}
s_{i}(x)\triangleq &a(x)(1+x^\tau) \bmod (1+x^{p\tau})\\
=&(a_0+a_{(p-1)\tau})+(a_1+a_{(p-1)\tau+1})x+\cdots+\\
&(a_{\tau}+a_{0})x^{\tau}+\cdots+(a_{p\tau-1}+a_{(p-1)\tau-1})x^{p\tau-1}.
\end{align*}
Therefore, we obtain
\begin{align*}
s_{\mu,i}=&a_{\mu}+a_{(p-1)\tau+\mu},\\
s_{\tau+\mu,i}=&a_{\tau+\mu}+a_{\mu},\\
&\cdots,\\
s_{(p-1)\tau+\mu,i}=&a_{(p-1)\tau+\mu}+a_{(p-2)\tau+\mu},
\end{align*}
for $\mu=0,1,\ldots,\tau-1$. We can check that 
\begin{align*}
&s_{\mu,i}+s_{\tau+\mu,i}+\cdots+s_{(p-2)\tau+\mu,i}\\
=&(a_{\mu}+a_{(p-1)\tau+\mu})+(a_{\tau+\mu}+a_{\mu})+\cdots+\\
&(a_{(p-2)\tau+\mu}+a_{(p-3)\tau+\mu})\\
=&a_{(p-1)\tau+\mu}+a_{(p-2)\tau+\mu}\\
=&s_{(p-1)\tau+\mu,i}.
\end{align*}
Therefore, the coefficients of $s_{i}(x)$ satisfy \eqref{eq:check1}.
\end{proof}

Since  the equation 
\[
(1+x^{\tau})(x^{\tau}+x^{3\tau}+\cdots+x^{(p-2)\tau})+1\cdot h(x)=1
\]
holds over $\mathbb{F}_2[x]$, $1+x^{p\tau}$ can be factorized as a product of two co-prime factors $1+x^{\tau}$ and $h(x)$. We show in the next lemma that $R_{p\tau}$ is isomorphic to $\mathbb{F}_2[x]/(1+x^{\tau})\times \mathbb{F}_2[x]/(h(x))$. 
\begin{lemma}
The ring $R_{p\tau}$ is isomorphic to the product ring $\mathbb{F}_2[x]/(1+x^{\tau})\times \mathbb{F}_2[x]/(h(x))$.
\label{lm:ismp1}
\end{lemma}
\begin{proof}
We need to find an isomorphism between $R_{p\tau}$ and $\mathbb{F}_2[x]/(1+x^{\tau})\times \mathbb{F}_2[x]/(h(x))$. Indeed, we can set up an isomorphism
$$\theta: R_{p\tau} \rightarrow \mathbb{F}_2[x]/(1+x^{\tau}) \times \mathbb{F}_2[x]/(h(x))$$
by defining
\[
 \theta(f(x)) \triangleq (f(x) \bmod 1+x^{\tau}, f(x) \bmod h(x)).
\]
The mapping $\theta$ is a ring homomorphism and a bijection, because it has an inverse function $\phi((a(x),b(x)))$  given by
\[
\phi((a(x),b(x))) \triangleq [a(x) \cdot h(x) + b(x) \cdot e(x)] \bmod 1+x^{p\tau}.
\]
In the following, we show that the composition $\phi \circ \theta$ is the identity map of~$R_{p\tau}$.

For any polynomial $f(x)\in R_{p\tau}$, there exist two polynomials $g_1(x),g_2(x)\in \mathbb{F}_2[x]$ such that 
\begin{align*}
f(x)=& g_1(x)(1+x^\tau)+(f(x)\bmod (1+x^\tau)), \\
f(x)=& g_2(x)h(x)+(f(x)\bmod h(x)).
\end{align*}
Then we have
\begin{align*}
\tiny
\phi(\theta(f(x)))
=&\big(h(x)(f(x)\bmod (1+x^\tau))+\\
&e(x)(f(x)\bmod h(x))\big)\bmod 1+x^{p\tau}\\
=&\big(h(x)(f(x)-g_1(x)(1+x^\tau))+\\
&(1+h(x))(f(x)-g_2(x)h(x))\big)\bmod 1+x^{p\tau} \\
=&\big(h(x)f(x)-h(x)g_1(x)(1+x^\tau)+f(x)+\\
&h(x)f(x)-e(x)g_2(x)h(x)\big)\bmod 1+x^{p\tau} \\
=&\big(f(x)-h(x)g_1(x)(1+x^\tau)-\\
&e(x)g_2(x)h(x)\big)\bmod 1+x^{p\tau} \\
=&\big(f(x)-(1+x^{\tau})(x^{\tau}+x^{3\tau}+\cdots+x^{(p-2)\tau})\cdot \\
&g_2(x)h(x)\big)\bmod 1+x^{p\tau} \\
=&f(x).
\end{align*}
The composition $\phi \circ \theta$ is thus the identity mapping of~$R_{p\tau}$ and the lemma is proved.
\end{proof}

By Lemma \ref{lm:ismp1}, we have that $C_{p\tau}$ is isomorphic to $\mathbb{F}_2[x]/(h(x))$ in the next lemma. 
\begin{lemma}
The ring $C_{p\tau}$ is isomorphic to $\mathbb{F}_2[x]/(h(x))$. Furthermore, the  isomorphism
$\theta: C_{p\tau} \rightarrow \mathbb{F}_2[x]/(h(x))$
can be defined as
 $\theta(f(x)) \triangleq  f(x) \bmod h(x).$

\label{lm:ismp}
\end{lemma}
For example, when $p=5$ and $\tau=2$, $C_{10}$ is isomorphic to the ring $\mathbb{F}_2[x]/(1+x^2+x^4+x^6+x^{8})$ and the element $1+x^8$ in $C_{10}$ is mapped to
\[1+x^8 \bmod (1+x^2+x^4+x^6+x^{8})= x^2+x^4+x^6.
\]
If we apply the function $\phi$ to $x^2+x^4+x^6$, we can recover
\begin{align*}
\phi (0,x^2+x^4+x^6) =& (x^2+x^4+x^6)(x^2+x^4+x^6+x^8) \\
=& 1+x^8 \bmod (1+x^{10}).
\end{align*}

When $\tau=1$, the ring $C_{p}$ is discussed in \cite{shumregenerating,Hou2016BASIC} and used in a new class of regenerating codes with low computational complexity. Note that $C_{p\tau}$ is isomorphic to a finite field $\mathbb{F}_{2^{(p-1)\tau}}$ if and only if 2 is a primitive element in $\mathbb{Z}_p$ and $\tau=p^i$ for some non-negative integer $i$~\cite{Itoh1991Characterization}.

We need the following definition about $e(x)$-inverse before introducing the explicit constructions of the proposed array codes.

\begin{defn}
A polynomial $f(x)\in R_{p\tau}$ is called \emph{$e(x)$-invertible} if we can find a polynomial $\bar{f}(x)\in R_{p\tau}$ such that $f(x)\bar{f}(x)=e(x)$, where $e(x)$ is given in~\eqref{eq:identity-1}. The polynomial $\bar{f}(x)$ is called $e(x)$-inverse of $f(x)$.
\end{defn}

We show in the next lemma that $1+x^b$ is $e(x)$-invertible in $R_{p\tau}$.

\begin{lemma}
Let $b$ be an integer with $1\leq b<p\tau$ and the greatest common divisor (GCD) of $b$ and $p$ is $\gcd (b,p)=1$, and let $\gcd(b,\tau)=a$. The $e(x)$-inverse of $1+x^b$ in $R_{p\tau}$ is 
\begin{eqnarray}
\sum_{i=\tau/a}^{2\tau/a-1}x^{ib}+\sum_{i=3\tau/a}^{4\tau/a-1}x^{ib}+\cdots+\sum_{i=(p-2)\tau/a}^{(p-1)\tau/a-1}x^{ib}.
\label{eq:inverse}
\end{eqnarray}
\label{lm:inv}
\end{lemma}
\begin{proof}
We can check that, in  $R_{p\tau}$,
\begin{align*}
&(1+x^b)\Big((x^{\frac{\tau}{a}b}+x^{(\frac{\tau}{a}+1)b}+\cdots+x^{(2\frac{\tau}{a}-1)b})+\\
&(x^{3\frac{\tau}{a}b}+x^{(3\frac{\tau}{a}+1)b}+\cdots+x^{(4\frac{\tau}{a}-1)b})+\cdots+\\
&(x^{(p-2)\frac{\tau}{a}b}+x^{((p-2)\frac{\tau}{a}+1)b}+\cdots+x^{((p-1)\frac{\tau}{a}-1)b})\Big)\\
=&(x^{\frac{\tau}{a}b}+x^{2\frac{\tau}{a}b}+x^{3\frac{\tau}{a}b}+\cdots+x^{(p-2)\frac{\tau}{a}b}+x^{(p-1)\frac{\tau}{a}b}).
\end{align*}
It is sufficient to show that the above equation is equal to $e(x)$, i.e., 
\begin{equation}
\begin{array}{ll}
&x^{\frac{\tau}{a}b}+x^{2\frac{\tau}{a}b}+x^{3\frac{\tau}{a}b}+\cdots+x^{(p-2)\frac{\tau}{a}b}+x^{(p-1)\frac{\tau}{a}b}\\
\equiv & e(x) \bmod (1+x^{p\tau}).
\end{array}
\label{eq:ba}
\end{equation}
Consider the ring of integers modulo $p\tau$, which is denoted $\mathbb{Z}_{p\tau}$. In $\mathbb{Z}_{p\tau}$, there is a set 
\[\tau\mathbb{Z}_{p\tau}\triangleq \left( 0, \tau, 2\tau, \ldots, (p-1)\tau\right).\]
Now we consider $(i b/a)\bmod p\tau\in \mathbb{Z}_{p\tau}$ for $i\in\{1,2,\ldots,p-1\}$. Therefore,
$$(\tau i b/a)\bmod p\tau \in \tau\mathbb{Z}_{p\tau}.$$
Next, we want to show that, for $i\neq j\in\{1,2,\ldots,p-1\}$,
\[i\tau b/a \neq j\tau b/a \bmod p\tau.\]
Assume that $i\tau b/a \bmod p\tau= j\tau b/a \bmod p\tau$. Then there exists an integer $\ell$ such that
$$i\tau b/a=\ell p\tau+j\tau b/a.$$ The above equation can be further reduced to
$$(i-j) b/a=\ell p.$$ Since $\gcd(b,p)=1$, $\gcd(b/a,p)=1$. Hence, we have $p|(i-j)$. However, this is impossible due to the fact that  $1\le j<i\le p-1$. Similarly, we can prove that, for $1\le i\le p-1$,
$$i\tau b/a \bmod p\tau \neq 0.$$ 
Hence, we can obtain that 
\[
\left(\tau b/a, \ldots, (p-1)\tau b/a\right)\equiv \left(\tau, \ldots, (p-1)\tau\right)\bmod p\tau.
\]
Therefore,  \eqref{eq:ba} holds.
\end{proof}

By Lemma~\ref{even-parity}, we have $s_i(x)\in C_{p\tau}$ for $i=1,2,\ldots, k$. Let $f(x)$ be any entry of the generator matrix or check matrix.  If $f(x)\notin C_{p\tau}$, one may replace the entry by $(f(x)e(x) \bmod (1+x^{p\tau}))\in C_{p\tau}$ without changing the results. This is due to the fact that $s_i(x)e(x)=s_i(x)\bmod (1+x^{p\tau})$, for $1\le i\le k$.
Hence, after replacing all  $f(x)$ of the generator matrix or check matrix with $(f(x)e(x) \bmod (1+x^{p\tau}))$, we have an equivalent generator matrix or check matrix  such that the coded polynomials in~\eqref{eq:data_clm} can be computed over the ring $C_{p\tau}$ via \eqref{eq:encoding} or \eqref{eq:check}. 

The encoding procedure can be described in terms of polynomial operations  as follows. Given $k(p-1)\tau$ information bits, by \eqref{eq:poly}, one appends $\tau$ extra bits for each of $(p-1)\tau$ information bits and forms $k$ data polynomials that  belong to  $C_{p\tau}$.
After obtaining the vector in \eqref{eq:data_clm} by choosing some specific encoding matrix or check matrix, one stores the coefficients  in the polynomials of degrees $0$ to $(p-1)\tau-1$ and drops the coefficients in higher degrees.
The proposed array code can be considered as punctured systematic linear code over $C_{p\tau}$.

\section{Two Explicit Constructions of Binary MDS Array Codes}
\label{sec:construction}
The purpose of this paper is to find suitable encoding matrices $\mathbf{P}_{k\times r}$ or check matrices $\mathbf{H}_{r\times (k+r)}$ such that the corresponding codes are MDS codes and the repair bandwidth of one single failure is asymptotically weak-optimal. In the section, we will give two explicit constructions of  binary MDS array codes, where the first construction is formed by an encoding matrix and the second construction is formed by a check matrix. The repair bandwidth of the first construction is asymptotically weak-optimal for any one single information failure, while the repair bandwidth of the second construction is asymptotically weak-optimal for any one single information or parity failure.

Note that, for codes  constructed from both constructions, not all parameters exist for them to be MDS codes. A sufficient condition will be derived for them to be MDS codes when $p$ is large enough.  Some constructed codes with small $p$ are also proved to be MDS codes. 
\vspace{-0.5cm}
\subsection{The First Construction: Encoding Matrix}
\label{sec:constr1}
Let $\eta=d-k+1$,\footnote{$\eta$ and $d-k+1$ will be interchangeably used in the work.} $k\geq 4$, $r\geq 3$ be an odd number, $d=k+(r-1)/2$, $\tau=(d-k+1)^{k-2}$, and $p>d-k$.
The constructed code is denoted by $\mathcal{C}_1(k,r,d,p)$ with $\mathbf{P}_{k\times r}$ given in \eqref{matrixp}.  
\newcounter{mytempeqncnt}
\begin{figure*}[!t]
\normalsize
\begin{equation}
\mathbf{P}_{k\times r}\triangleq 
\left[\begin{array}{ccccc|cccc}
 1& x&x^2 & \cdots & x^{d-k} &1 &1 & \cdots &1\\
 1 & x^{\eta} & x^{2\eta}&  \cdots & x^{(d-k)\eta}& x^{(d-k)\eta^{k-2}}& x^{(d-k-1)\eta^{k-2}} &\cdots &x^{\eta^{k-2}} \\
 1 & x^{\eta^2} & x^{2\eta^2}& \cdots & x^{(d-k)\eta^{2}}& x^{(d-k)\eta^{k-3}}&x^{(d-k-1)\eta^{k-3}}& \cdots &x^{\eta^{k-3}} \\
 \vdots &\vdots&\vdots&\ddots&\vdots&\vdots&\vdots&\ddots&\vdots\\
 1 & x^{\eta^{k-3}} & x^{2\eta^{k-3}}& \cdots & x^{(d-k)\eta^{k-3}}& x^{(d-k)\eta^{2}}&x^{(d-k-1)\eta^{2}}&\cdots &x^{\eta^{2}} \\
 1 &x^{\eta^{k-2}}  & x^{2\eta^{k-2}}& \cdots & x^{(d-k)\eta^{k-2}}& x^{(d-k)\eta}&x^{(d-k-1)\eta}&\cdots &x^{\eta} \\
 1 &1  & 1& \cdots & 1& x^{d-k}&x^{d-k-1}& \cdots &x 
 \end{array}\right].
\label{matrixp}
\end{equation}
\end{figure*}
Since every data polynomial is in $C_{p\tau}$ and $C_{p\tau}$ is an ideal, we have the following lemma.
\begin{lemma}
For $j=k+1,k+2,\ldots,k+r$, each coded polynomial $s_{j}(x)$ in $\mathcal{C}_1(k,r,d,p)$ belongs to $C_{p\tau}$.
\label{lm:codedC}
\end{lemma}
By Lemma \ref{even-parity} and Lemma \ref{lm:codedC}, the coefficients of the coded polynomials $s_{j}(x)$ satisfy \eqref{eq:check1} if we replace $i$ with $j$ in \eqref{eq:check1}.
Let $(i:j)=\{i,i+1,\ldots,j\}$ and $\mathbf{P}_{k\times r}(i:j)$ be the sub-matrix of $\mathbf{P}_{k\times r}$ with column index  determined by  $(i:j)$. In $\mathbf{P}_{k\times r}$, the sub-matrix $\mathbf{P}_{k\times r}(\eta+1:2\eta-1)$ is a clockwise rotation of the sub-matrix $\mathbf{P}_{k\times r}(2:\eta)$ by 180 degrees. The last row of $\mathbf{P}_{k\times r}(2:\eta)$ is an all one vector, and the exponent of the entry in row $i$ and column $j$ of $\mathbf{P}_{k\times r}(2:\eta)$ is ${\eta^{i-1}}$ times of that in the first row and column $j$ for $i=2,3,\ldots,k-1$ and $j=1,2,\ldots, d-k$.

\begin{example}
\label{example1}
Consider  $k=4, p=3, r=3$. Hence, $d=4+1=5$ and $\tau=4$. The 32 information bits are represented by $s_{0,i}$, $s_{1,i}, \ldots, s_{7,i}$, for $i=1,2,3,4$. The encoding matrix is
\begin{align*}
\mathbf{P}_{4\times 3}=\begin{bmatrix}
1 & x & 1 \\
1 & x^2 & x^4 \\
1 & x^4 & x^2 \\
1&1&x
\end{bmatrix}.
\end{align*}
The Example~\ref{example1} is illustrated in Fig.~\ref{example}.
Note that the extra bits calculated from the information bits do not need to be stored and they are only used to calculate the parity bits.

\begin{figure*}
\centering
\includegraphics[width=0.79\textwidth]{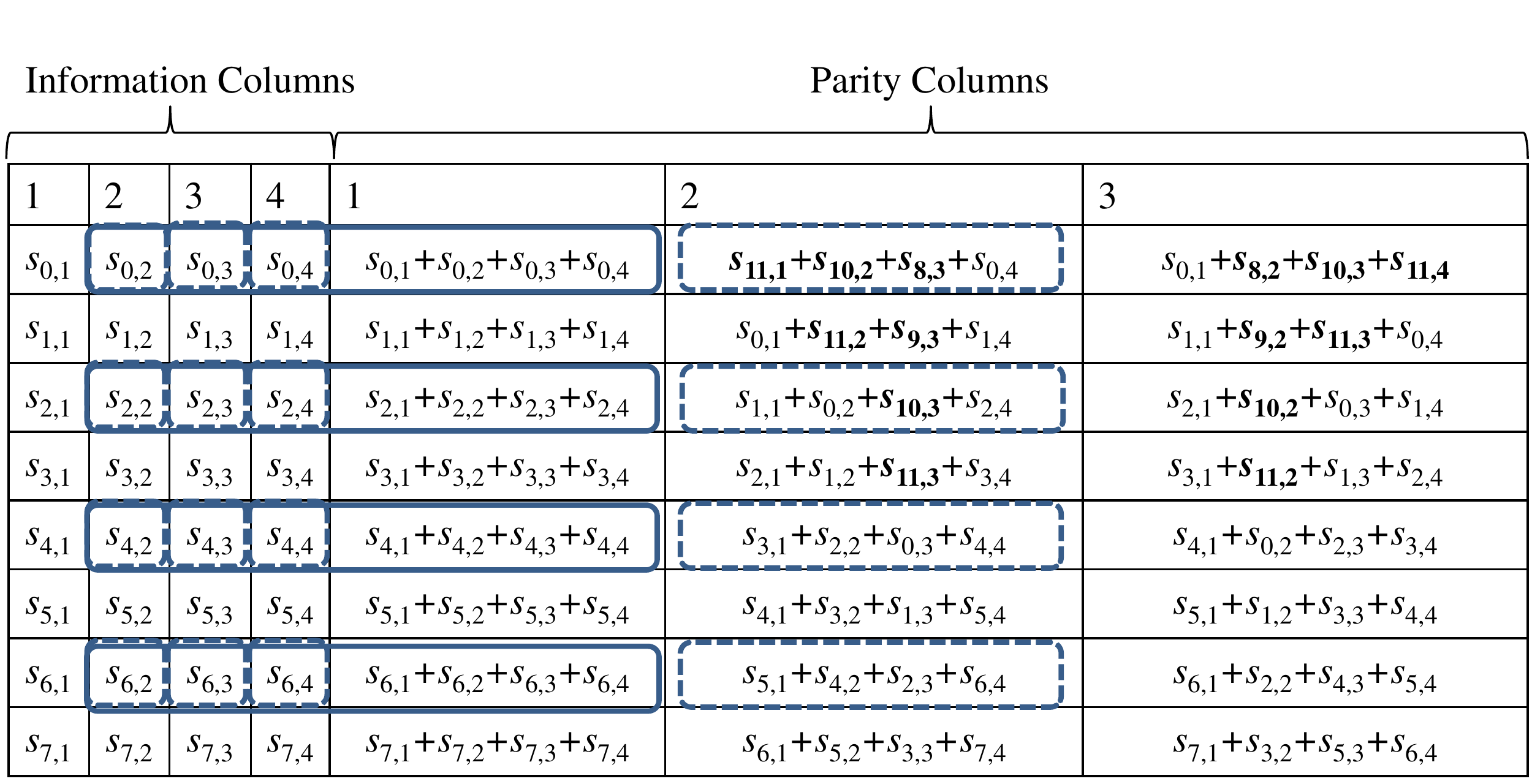}
\caption{The information and parity columns in Example~\ref{example1}. When information column 1 fails, the bits in the solid line box are downloaded to repair the information bits $s_{0,1},s_{2,1},s_{4,1},s_{6,1}$ and the bits in the dashed box are used to repair the information bits $s_{1,1},s_{3,1},s_{5,1},s_{7,1}$.}
\label{example}
\end{figure*}

Suppose that the first column fails. We can repair the bits $s_{0,1},s_{2,1},s_{4,1},s_{6,1}$ by
\begin{eqnarray*}
& s_{0,1}=s_{0,2}+s_{0,3}+s_{0,4}+(s_{0,1}+s_{0,2}+s_{0,3}+s_{0,4}) \\
& s_{2,1}=s_{2,2}+s_{2,3}+s_{2,4}+(s_{2,1}+s_{2,2}+s_{2,3}+s_{2,4}) \\
& s_{4,1}=s_{4,2}+s_{4,3}+s_{4,4}+(s_{4,1}+s_{4,2}+s_{4,3}+s_{4,4}) \\
& s_{6,1}=s_{6,2}+s_{6,3}+s_{6,4}+(s_{6,1}+s_{6,2}+s_{6,3}+s_{6,4}),
\end{eqnarray*}
and repair the bits $s_{1,1},s_{3,1},s_{5,1},s_{7,1}$ by
\begin{align*}
s_{1,1}=&s_{0,2}+\mathbf{s_{10,3}}+s_{2,4}+(s_{1,1}+s_{0,2}+s_{10,3}+s_{2,4}) \\
s_{3,1}=&s_{2,2}+s_{0,3}+s_{4,4}+(s_{3,1}+s_{2,2}+s_{0,3}+s_{4,4}) \\
s_{5,1}=&s_{4,2}+s_{2,3}+s_{6,4}+(s_{5,1}+s_{4,2}+s_{2,3}+s_{6,4}) \\
s_{7,1}=&s_{6,2}+s_{4,3}+\mathbf{s_{8,4}}+(\mathbf{s_{11,1}}+\mathbf{s_{10,2}}+\mathbf{s_{8,3}}+s_{0,4})\\
&+(s_{3,1}+s_{2,2}+s_{0,3}+s_{4,4}).
\end{align*}

As we can compute $\mathbf{s_{10,3}}$ by $s_{6,3}+s_{2,3}$ and $\mathbf{s_{8,4}}$ by $s_{4,4}+s_{0,4}$, we do not need to download the bits $\mathbf{s_{10,3}}$ and $\mathbf{s_{8,4}}$.
Therefore,  we only need to download 4 bits from each of three information columns and two parity columns. There are total 20 bits downloaded from 5 columns to repair the bits of the first information column.
Namely, only half of the bits of the helping columns are accessed. In Fig.~\ref{example}, the bits in the solid line box are downloaded to repair the information bits $s_{0,1},s_{2,1},s_{4,1},s_{6,1}$ and the bits in the dashed box are used to repair the information bits $s_{1,1},s_{3,1},s_{5,1},s_{7,1}$.

Note that, according to Theorem \ref{thm:MDSs}, the code given in Fig.~\ref{example} is not an MDS code due to the fact that there exists a $2\times 2$ sub-matrix whose determinant is divisible by $1+x+x^2$.
\end{example}

\subsection{The Second Construction: Check Matrix}
\label{sec:constr2}
Let $k\geq 4$, $r\geq 4$ be an even number, the matrix $\mathbf{0}$ be the $\frac{r}{2}\times \frac{r}{2}$ zero matrix, $\eta=d-k+1=\frac{r}{2}$, $\tau=\eta^{d-1}=(\frac{r}{2})^{d-1}$, and $p>\eta$. The second construction is denoted by $\mathcal{C}_2(k,r,d,p)$ with the check matrix $\mathbf{H}_{r\times (k+r)}$ defined in \eqref{matrixh} in the next page.
\begin{figure*}[!t]
\normalsize
\begin{equation}
\left[\begin{array}{ccccccccc}
 1&x&x^2 & \cdots & x^{\eta-1} & & & &\\
 1 & x^{\eta} & x^{2\eta}& \cdots &  x^{(\eta-1)\eta}& & \text{\huge 0}& &\\
 \vdots & \vdots & \vdots& \ddots &  \vdots & & & &\\
 1 & x^{(\eta)^{\eta-1}} & x^{2(\eta)^{\eta-1}}& \cdots &  x^{(\eta-1)(\eta)^{\eta-1}}& & & &\\ \hline
 \vdots &\vdots&\vdots&\ddots&\vdots& 1 &\cdots &1 & 1\\
 1 & x^{(\eta)^{d-1}} & x^{2(\eta)^{d-1}}& \cdots & x^{(\eta-1)(\eta)^{d-1}}& x^{(\eta-1)(\eta)^{d-1}}&\cdots &x^{(\eta)^{d-1}} &x^{(d-1)(\eta)^{d-1}}\\
 1 & 1 & 1& \cdots & 1&\vdots &\ddots &\vdots &\vdots\\ \hline
  &  & &  & &x^{(\eta-1)(\eta)^{\eta-1}} &\cdots &x^{(\eta)^{\eta-1}} &x^{(\eta-1)(\eta)^{d-1}}\\
 &  & \text{\huge 0}&  & &\vdots  &\ddots &\vdots &\vdots \\
  &  & &  & &x^{(\eta-1)\eta} &\cdots &x^{\eta} &x^{(\eta)^{d-1}}\\
  &  & &  & &x^{(\eta-1)} &\cdots &x &1
   \end{array}\right]^T.
\label{matrixh}
\end{equation}
\end{figure*}
Unlike the definition given in~\eqref{eq:data_clm}, it will be shown later that, for $\mathcal{C}_2(k,r,d,p)$,   $k$ data polynomials can be placed in any $k$ positions of the codewords. To compute the coded polynomials, we should solve a linear equation system with the encoding coefficients being a sub-matrix of~\eqref{matrixh}. 

In the following, we first give a construction for $r=4$. The cases for  $r>4$ will be given later.  A $\mathcal{C}_2(k,4,d,p)$ contains $k+4$ polynomials $s_1(x),s_2(x),\ldots,s_{k+4}(x)$, where we select $s_3(x),s_4(x),\ldots,s_{k+2}(x)$ as data polynomials and $s_1(x),s_2(x),s_{k+3}(x),s_{k+4}(x)$ as coded polynomials. 
The check matrix $\mathbf{H}_{4\times (k+4)}$ is given as follows: 
\begin{equation}
\small
\left[\begin{array}{cc|ccccc|cc}
1 & 1 & 1 & 1 & \cdots & 1 & 1 & 0 & 0\\
x & x^2 & x^4 & x^8 & \cdots & x^{2^k} & 1 & 0 &0\\ \hline
0 & 0 & 1 & x^{2^k} & \cdots & x^8 & x^4 & x^2 & x\\
0 & 0 & 1 & x^{k\cdot 2^k} & \cdots & x^{3\cdot 2^k} & x^{2\cdot 2^k} & x^{2^k} & 1\\
\end{array}\right].
\label{check4}
\end{equation}
From the first two rows of $\mathbf{H}_{4\times (k+4)}$, we have 
\begin{align}
& s_1(x)+s_2(x)+\cdots+s_{k+1}(x)+s_{k+2}(x)=0, \label{eq:s_{k+2}-1}\\
& xs_1(x)+x^2s_2(x)+\cdots+x^{2^k}s_{k+1}(x)+s_{k+2}(x)=0. \label{eq:s_{k+2}-2}
\end{align}
First, we can compute the summation of cyclic-shifted version of the data polynomials as
\begin{align*}
& p_{1}(x)\triangleq  s_3(x)+s_4(x)+\cdots+s_{k+1}(x)+s_{k+2}(x), \\
& p_2(x)\triangleq x^4s_3(x)+x^8s_4(x)+\cdots+x^{2^k}s_{k+1}(x)+s_{k+2}(x). 
\end{align*}
Substituting \eqref{eq:s_{k+2}-1} into $ p_{1}(x)$ and \eqref{eq:s_{k+2}-2} into $ p_{2}(x)$,
we have 
\[
xs_2(x)=\frac{xp_1(x)+p_2(x)}{1+x}.
\]
Note that we may solve for $s_2(x)$ by left cyclic shifting of $xs_2(x)$. It is easy to see that $p_1(x),p_2(x)\in C_{p\tau}$, and therefore $xp_1(x)+p_2(x)\in C_{p\tau}$. 
By Lemma \ref{lm:inv}, $1+x^b$ is $e(x)$-invertible and we can compute $g(x)$ from $(1+x^b)g(x)=f(x)$ as
\begin{eqnarray}
g(x)=f(x)\Big(\sum_{i=\frac{\tau}{a}}^{\frac{2\tau}{a}-1}x^{ib}+\sum_{i=\frac{3\tau}{a}}^{\frac{4\tau}{a}-1}x^{ib}+\cdots+\sum_{i=\frac{(p-2)\tau}{a}}^{\frac{(p-1)\tau}{a}-1}x^{ib}\Big),
\label{eq:inverse1}
\end{eqnarray}
where $f(x)\in C_{p\tau}$ and $\gcd(b,\tau)=a$. By letting $b=1$ and $f(x)=xp_1(x)+p_2(x)$, we can solve $g(x)=xs_2(x)$ by \eqref{eq:inverse1} and then $s_1(x)$ by~\eqref{eq:s_{k+2}-1}. 
As $f(x)\in C_{p\tau}$, the resulting polynomial $g(x)$ in \eqref{eq:inverse1} is also in $C_{p\tau}$. Therefore, the coded polynomials $s_1(x),s_2(x)$ are in $C_{p\tau}$. The other two coded polynomials $s_{k+3}(x),s_{k+4}(x)$ can be computed similarly and are also in $C_{p\tau}$.

Next, we demonstrate how to efficiently compute $g(x)$ via~\eqref{eq:inverse1}. We can compute the coefficients $g_j$ for $j=0,1,\ldots,a-1$ by
\begin{eqnarray}
g_j=&\sum_{i=\frac{\tau}{a}}^{\frac{2\tau}{a}-1}f_{(j-ib)\bmod p\tau}+\sum_{i=\frac{3\tau}{a}}^{\frac{4\tau}{a}-1}f_{(j-ib)\bmod p\tau}+ \nonumber\\
&\cdots+\sum_{i=\frac{(p-2)\tau}{a}}^{\frac{(p-1)\tau}{a}-1}f_{(j-ib)\bmod p\tau}.
\label{eq:inverse2}
\end{eqnarray}
Since $(1+x^b)g(x)=f(x)$, once $g_0,g_1,\ldots,g_{a-1}$ are known, we can compute the other coefficients of $g(x)$ iteratively by
\begin{equation}
g_{b\ell+j}=f_{b\ell+j}+g_{b(\ell-1)+j}
\label{eq:itera}
\end{equation}
with the index $\ell$ running from 1 to $p\tau/a-1$ and $0\le j\le a-1$. It can be shown that $b\ell_1+j\bmod p\tau \neq b\ell_2+j\bmod p\tau$ for $\ell_1\neq \ell_2 \in\{1,2,\ldots,p\tau/a-1\}$. Therefore, we can compute all the other coefficients of $g(x)$ by \eqref{eq:itera}.
We can count that there are $$a(\frac{p-1}{2}\cdot \frac{\tau}{a}-1)+(p\tau-a)=\frac{3p\tau-\tau-4a}{2}$$ XORs involved in solving $g(x)$ from $(1+x^b)g(x)=f(x)$. 

\begin{example}
\label{example2}	

We give a construction  with $k=4, p=3, r=4$, and then $d=5$ and $\tau=16$. Each column stores $(p-1)(\frac{r}{2})^{d-1}=32$ bits and there are $k+r=8$ columns. We have 8 polynomials $s_i(x)=\sum_{\ell=0}^{47}s_{\ell,i}x^\ell$ for $i=1,2,\ldots,8$. Suppose that $s_3(x),s_4(x),s_5(x),s_6(x)$ are four data polynomials. 
First, we compute $p_1(x)=\sum_{\ell=0}^{47}p_{\ell,1}x^\ell,p_2(x)=\sum_{\ell=0}^{47}p_{\ell,2}x^\ell$ as
\begin{align*}
p_1(x)=&s_3(x)+s_4(x)+s_5(x)+s_6(x),\\
p_2(x)=&x^4s_3(x)+x^8s_4(x)+x^{16}s_5(x)+s_6(x).
\end{align*}
Then, we can compute $xs_2(x)$ via
\[
xs_2(x)=\frac{x p_1(x)+p_2(x)}{1+x}
\]
by \eqref{eq:inverse2} and \eqref{eq:itera}, and obtain $s_2(x)$ by left cyclic shifting of $xs_2(x)$, i.e.,
\begin{equation*}
s_{\ell,2}=\left\{\begin{array}{ll}
&\sum_{j=0}^{\ell}(p_{j,1}+p_{j+1,2})+\\
&\sum_{i=16}^{31}(p_{47-i,1}+p_{48-i,2}) \text{ for } \ell=0,1,\ldots,46;\\
&\sum_{i=16}^{31}(p_{47-i,1}+p_{48-i,2}) \text{ for } \ell=47.
\end{array}\right.
\end{equation*}
We can then compute $s_1(x)$ as $s_1(x)=s_2(x)+p_1(x)$, i.e., 
\begin{equation*}
s_{\ell,1}=\left\{\begin{array}{ll}
&p_{\ell,1}+\sum_{j=0}^{\ell}(p_{j,1}+p_{j+1,2})+\\
&\sum_{i=16}^{31}(p_{47-i,1}+p_{48-i,2}) \text{ for } \ell=0,1,\ldots,46;\\
&p_{\ell,1}+\sum_{i=16}^{31}(p_{47-i,1}+p_{48-i,2}) \text{ for } \ell=47.
\end{array}\right.
\end{equation*}
We can also compute columns 7 and 8 in a similar way. 

Suppose that the first column is failed. We can recover the bits $s_{2\ell,1}$ for $\ell=0,1,\ldots,15$ by 
\begin{eqnarray}
s_{2\ell,1}=s_{2\ell,2} + s_{2\ell,3} + s_{2\ell,4} + s_{2\ell,5} + s_{2\ell,6}
\label{eq:example2rep}
\end{eqnarray}
and the other bits $s_{2\ell+1,1}$ for $\ell=0,1,\ldots,15$ by 
\begin{eqnarray}
s_{2\ell+1,1}=s_{2\ell,2} + s_{2\ell-2,3} + s_{2\ell-6,4} + s_{2\ell-14,5} + s_{2\ell+2,6}.
\label{eq:example2rep1}
\end{eqnarray}
Note that we do not need to download the extra bits ${\bf s_{46,3}}$, ${\bf s_{48-2\ell,4}}$ for $\ell=1,2,3$, ${\bf s_{48-2\ell,5}}$ for $\ell=1,2,\ldots,7$ and ${\bf s_{32,6}}$ as we can compute them by summing two downloaded bits in \eqref{eq:example2rep}. The other bits needed in \eqref{eq:example2rep1} are also needed in \eqref{eq:example2rep}. Therefore, we only need to download 16 bits $s_{2\ell,i}$ for $\ell=0,1,\ldots,15$ from columns $i=2,3,4,5,6$ and there are total $80$ bits downloaded from 5 columns to repair the first column.

Note that, according to Theorem \ref{thm:MDSs}, the code given in Example~\ref{example2} is not an MDS code due to the fact that there exists a sub-matrix whose determinant is divisible by $1+x+x^2$.

\end{example}

The encoding procedure with  $r>4$ is described as follows. Since we will show that $\mathcal{C}_2(k,r,d,p)$ satisfy the MDS condition in Theorem \ref{thm:7} (in Section \ref{sec:MDS}), the encoding procedure can be implemented as a special case of decoding procedure. There are $k+r$ polynomials $s_1(x),s_2(x),\ldots,s_{k+r}(x)$, and assume that $k$ data polynomials are $s_{\frac{r}{2}+1},s_{\frac{r}{2}+2},\ldots,s_{\frac{r}{2}+k}$. We first replace each entry $f(x)$ of $\mathbf{H}_{r\times (k+r)}$ with $f(x)e(x)\bmod (1+x^{p\tau})$, and then solve the $r$ linear equations for $r$ coded polynomials according to the modified check matrix over $C_{p\tau}$. As $\mathcal{C}_2(k,r,d,p)$ satisfy the MDS condition, we can always compute the $r$ coded polynomials.

\section{The MDS Property}
\label{sec:MDS}
Let $\bar{\mathbf{G}}_{k\times (k+r)}$ and $\bar{\mathbf{H}}_{r\times (k+r)}$ be the matrices by replacing each entry $f(x)$ of $\mathbf{G}_{k\times (k+r)}$ and $\mathbf{H}_{r\times (k+r)}$ with $f(x)e(x)\bmod (1+x^{p\tau})$ respectively. The codes satisfy MDS property if and only if the determinant of any $k\times k $ sub-matrix of $\bar{\mathbf{G}}_{k\times (k+r)}$ or $r\times r$ sub-matrix of $\bar{\mathbf{H}}_{r\times (k+r)}$ is $e(x)$-invertible. Recall that $C_{p\tau}$ is isomorphic to $\mathbb{F}_2[x]/(h(x))$ by Lemma \ref{lm:ismp}, the necessary and sufficient MDS condition is equivalent to that the determinant of any $k\times k $ sub-matrix of $\bar{\mathbf{G}}_{k\times (k+r)}$ or $r\times r$ sub-matrix of $\bar{\mathbf{H}}_{r\times (k+r)}$, after reducing modulo $h(x)$, is invertible in $\mathbb{F}_2[x]/(h(x))$. Note that the determinant of any $k\times k $ sub-matrix of $\bar{\mathbf{G}}_{k\times (k+r)}$ or $r\times r$ sub-matrix of $\bar{\mathbf{H}}_{r\times (k+r)}$ after  reducing modulo $h(x)$ can be computed by first reducing each entry of the square matrix by $h(x)$, and then computing the determinant after reducing modulo $h(x)$. For any integers $i,j$, we have
\[
(x^{ip\tau+j}\bmod (1+x^{p\tau}))\bmod h(x)=x^j \bmod h(x)
\] 
and
\begin{align*}
&(f(x)e(x)\bmod (1+x^{p\tau}))\bmod h(x)\\
=&(f(x)e(x)\bmod h(x))\bmod 1+x^{p\tau}\\
=&[f(x)(1+h(x))\bmod h(x)]\bmod 1+x^{p\tau}\\
=&(f(x)\bmod 1+x^{p\tau})\bmod h(x)\\
=&f(x)\bmod h(x).
\end{align*} 
It is sufficient to show that the determinant of any $k\times k $ sub-matrix of $\mathbf{G}_{k\times (k+r)}$ (or equivalently any square sub-matrix of \eqref{matrixp} by Corollary 3 and Theorem 8 in \cite{MacWilliamsSloane}) or $r\times r$ sub-matrix of $\mathbf{H}_{(k+r)\times r}$, after  modulo $h(x)$, is invertible over $\mathbb{F}_2[x]/(h(x))$.


\begin{theorem}
Let $h(x)$ be factorized as a product of powers of irreducible polynomials over $\mathbb{F}_2$:
\begin{equation}
h(x)\triangleq (f_1(x))^{\ell_1}\cdot (f_2(x))^{\ell_2}\cdots (f_t(x))^{\ell_t}, 
\label{eq:factor}
\end{equation}
where $\ell_i\geq 1$ for $i=1,2,\ldots,t$. $\mathcal{C}_1(k,r,d,p)$ (or $\mathcal{C}_2(k,r,d,p)$) is an MDS code if and only if the determinant of any $\ell\times \ell$ sub-matrix of \eqref{matrixp} for $\ell=1,2,\ldots,\min\{k,r\}$ (or the determinant of any $r\times r$ sub-matrix of~\eqref{matrixh}) is a non-zero polynomial in $\mathbb{F}_2[x]/(f_i(x))$ for $i=1,2,\ldots,t$. 
\label{thm:mds1}
\end{theorem}
\begin{proof}
By Chinese remainder theorem, the ring $\mathbb{F}_2[x]/h(x)$ is isomorphic to the direct sum of $t$ rings $\mathbb{F}_2[x]/(f_1(x))^{\ell_1}, \mathbb{F}_2[x]/(f_2(x))^{\ell_2}, \ldots, \mathbb{F}_2[x]/(f_t(x))^{\ell_t}$, and the mapping $\theta$ is defined by
\[
\theta (a(x))\triangleq ( a(x) \bmod (f_1(x))^{\ell_1},  \ldots, a(x) \bmod (f_t(x))^{\ell_t}),
\]
where $a(x) \in \mathbb{F}_2[x]/h(x)$. The inverse of $\theta$ is given by
\begin{align*}
&\theta^{-1}(a_1(x),\ldots,a_{t}(x))\triangleq \\
&\sum_{i=1}^{t}a_i(x)\cdot h_i(x) \cdot < h_i(x)^{-1} >_{f_i(x)^{\ell_i}} \bmod h(x),
\end{align*}
where $a_i(x) \in \mathbb{F}_2[x]/(f_i(x))^{\ell_i}$ and $h_i(x)=h(x)/(f_i(x))^{\ell_i}$ for $i=1,2,\ldots,t$, $< h_i(x)^{-1} >_{f_i(x)^{\ell_i}}$ denotes the multiplicative inverse of $h_i(x) \bmod f_i(x)^{\ell_i}$ and $1\leq i\leq t$. It can be checked that 
\[
\theta(\theta^{-1}(a_1(x),\ldots,a_{t}(x)))=(a_1(x),\ldots,a_{t}(x)).
\]
A more general version of the polynomial Chinese remainder theorem is presented in \cite[Theorem 1]{Fan2016A}.

The code is MDS if and only if the determinants of all the sub-matrices are invertible in $\mathbb{F}_2[x]/h(x)$. 
Suppose that a determinant $a(x)$ is invertible, i.e., there exists a polynomial $g(x)$ such that $g(x)a(x)=1 \bmod h(x)$. If we apply the mapping $\theta$ to $g(x)a(x)\bmod\ h(x)$, then we have 
\begin{align*}
&\theta (g(x)a(x)\bmod h(x))=\theta (1)\\
=&( 1 \bmod (f_1(x))^{\ell_1}, \ldots, 1 \bmod (f_t(x))^{\ell_t}).
\end{align*}
Therefore, $a(x)$ is invertible in $\mathbb{F}_2[x]/(f_i(x))^{\ell_i}$ and is a non-zero polynomial in $\mathbb{F}_2[x]/f_i(x)$ for $i=1, 2, \ldots, t$.

Conversely, suppose that a determinant $a(x)\bmod f_i(x)$ is a non-zero polynomial in $\mathbb{F}_2[x]/f_i(x)$ for $i=1, 2, \ldots, t$. As $f_i(x)$ is irreducible polynomial, we have $\gcd (f_i(x), a(x))=1$. Hence,  $\gcd (f_i(x)^{\ell_i}, a(x))=1$. Therefore, there exists a polynomial $g_i(x)$ such that $g_i(x)a(x)=1 \bmod f_i(x)^{\ell_i}$.  
By Chinese remainder theorem, there exists a unique polynomial $f(x)\in \mathbb{F}_2[x]/h(x)$ such that 
\begin{align*}
&\theta(f(x)a(x) \bmod h(x))\\
=&(g_1(x)a(x)\bmod f_1(x)^{\ell_1},\cdots, g_t(x)a(x)\bmod f_t(x)^{\ell_t})\\
=&(1,\cdots,1).
\end{align*}
By applying the inverse mapping $\theta^{-1}$ to $\theta(f(x)a(x) \bmod h(x))$, we obtain that 
\begin{align*}
&\theta^{-1}(\theta(f(x)a(x) \bmod h(x)))=f(x)a(x) \bmod h(x)\\
=&\theta^{-1}(g_1(x)a(x)\bmod (f_1(x))^{\ell_1},\ldots,g_{t}(x)a(x)\bmod (f_t(x))^{\ell_t})\\
=&\theta^{-1}(1,\ldots,1)
 =\sum_{i=1}^{t} h_i(x) \cdot < h_i(x)^{-1} >_{f_i(x)^{\ell_i}} \bmod\   h(x).
\end{align*}
As $\gcd ((f_i(x))^{\ell_i},(f_j(x))^{\ell_j})=1$ for $i\neq j$, we have  
\[
\gcd (h_1(x),h_2(x),\cdots,h_t(x))=1.
\]
By B$\acute{e}$zout's identity \cite{Area2006Extensions}, there exist polynomials $b_1(x),b_2(x),\ldots,b_t(x)$ in $\mathbb{F}_2[x]/h(x)$ such that 
\[
\sum_{i=1}^{t} h_i(x) \cdot b_i(x) \bmod h(x)=1
\]
and $b_i(x)$ can be computed as $b_i(x)=< h_i(x)^{-1} >_{f_i(x)^{\ell_i}}$. Therefore, we obtain that 
\begin{align*}
f(x)a(x) \bmod h(x)=&\sum_{i=1}^{t} h_i(x) \cdot b_i(x) \bmod h(x)\\
=&1 \bmod h(x),
\end{align*}
and $a(x)$ is invertible in $\mathbb{F}_2[x]/h(x)$. This completes the proof.
\end{proof}

A sufficient MDS condition is given in the next theorem.

\begin{theorem}
Let $h(x)$ be factorized as in \eqref{eq:factor}, where $\deg(f_1(x))\leq \deg(f_2(x))\leq \cdots \leq \deg(f_t(x))$. If $\deg(f_1(x))$ is larger than
\begin{equation}
\small
(r-1)(\frac{(\eta-1)((d-k)\eta^{k-1}-\eta^{k-\frac{r+1}{2}})-\eta^{k-1}+\eta^{k-\frac{r-1}{2}}}{(\eta-1)^2}),
\label{eq:suff1}
\end{equation}
then   $\mathcal{C}_1(k,r,d,p)$ is an MDS code for $k\geq r$. Similarly, $\mathcal{C}_2(k,r,d,p)$ is an MDS code for $k\geq r$, if $\deg(f_1(x))$ is larger than
\begin{equation}
(\eta-1)(\eta)^{d-1}-(\eta)^{d-\eta}-\frac{(\eta)^{d-1}-(\eta)^{d-\eta+1}}{\eta-1}.
\label{eq:suff2}
\end{equation}
\label{thm:7}
\end{theorem}
\begin{proof}
It is easy to see that \eqref{eq:suff1} is larger than $d-k$ and \eqref{eq:suff2} is larger than $\frac{r}{2}$. By Theorem~\ref{thm:mds1}, we should show that the determinants of all sub-matrices are invertible in $\mathbb{F}_2[x]/(f_i(x))$ for $i=1,2,\ldots,t$. If the maximum degree of the non-zero determinant is less than $\deg(f_1(x))$, then the determinant is invertible in $\mathbb{F}_2[x]/(f_i(x))$.

Since any $\ell\times \ell$, $1\le \ell\le r$, sub-matrix of \eqref{matrixp}  is contained in an $r\times r$ sub-matrix,  the maximum degree among all  determinants of $\ell\times \ell$ sub-matrices is no larger than that of all determinants of $r\times r$ sub-matrices. It is sufficient to calculate the maximum degree of the determinants of all $r\times r$ sub-matrices of \eqref{matrixp} for $\mathcal{C}_1(k,r,d,p)$.

Note that the size of the matrix \eqref{matrixp} is $k\times r$, we need to first choose $r$ rows from the $k$ rows to form an $r\times r$ sub-matrix and then calculate the maximum exponent of the determinant of the $r\times r$ sub-matrix. The determinant is computed as the summation (with plus or minus signs) of all possible multiplications of $r$ entries that are in different rows and different columns. Denote the row index and column index of $\ell$-th entry that is involved in computing the determinant with maximum degree among all the determinants as $i_{\ell}$ and $j_{\ell}$, respectively, where $\ell=1,2,\ldots,r$, $1\leq i_1<\cdots <i_r\leq k$ and $(j_1,j_2,\ldots,j_r)$ is a permutation of $(1,2,\ldots,r)$. There exists an integer $t$ for $1\leq t \leq r$ such that $j_{t}=d-k+1$.
If the maximum degree of all the determinants does not contain the entry $x^{(d-k)\eta^{k-2}}$ in column $d-k+1$ row $k-1$, i.e., $i_\ell\neq k-1$ and $j_\ell\neq d-k+1$ for all $\ell$, then we have that the exponent of the multiplication of $r$ entries with row indices 
\[
\{i_1,i_2,\ldots,i_{t-1},i_{t+1},\ldots,i_{r},k-1\}
\]
and column indices 
\[
\{j_1,j_2,\ldots,j_{t-1},j_{t+1},\ldots,j_{r},d-k+1\}
\]
is larger than the exponent of the determinant with row indices $i_1,i_2,\ldots,i_{r}$ and columns indices $j_1,j_2,\ldots,j_{r}$, as $(d-k)\eta^{k-2}$ is the largest exponent of all entries in \eqref{matrixp} and is larger than the exponent of all the entries in column $d-k+1$ except $x^{(d-k)\eta^{k-2}}$. This contradicts to the assumption that the determinant with row indices $i_1,i_2,\ldots,i_{r}$ and columns indices $j_1,j_2,\ldots,j_{r}$ has the maximum degree among all the determinants. Therefore, the maximum degree of all the determinants should contain the entry $x^{(d-k)\eta^{k-2}}$ in column $d-k+1$ and row $k-1$. Then, it is reduced to find $r-1$ entries that are in different columns and rows with each other in the $k\times r$ matrix \eqref{matrixp} except row $k-1$ and column $d-k+1$. By the same argument, we can obtain that the maximum degree of all the determinants contains the entry $x^{(d-k)\eta^{k-2}}$ in column $d-k+2$ and row $2$. Repeating the above procedure, we can obtain that the maximum degree of all determinants contains the entries in column $d-k+2-\ell$ and row $k-\ell$ for $\ell=1,2,\ldots,(r-1)/2+1$, and entries in column $d-k+1+\ell$ and row $1+\ell$ for $\ell=1,2,\ldots,(r-1)/2$.
It can be computed that the maximum degree is $2((d-k)\eta^{k-2}+(d-k-1)\eta^{k-3}+\cdots+\eta^{k-(r-1)/2-1})$, which is~\eqref{eq:suff1}. 
By repeating the same argument for \eqref{matrixh}, the maximum degree among all the determinants is achieved when the row indices of the $r\times r$ sub-matrix are $d-\frac{r}{2}-4$ to $d+\frac{r}{2}-5$ and the maximum degree is as given in~\eqref{eq:suff2}.
\end{proof}
By Theorem~\ref{thm:7}, we can choose $p$ with $\deg(f_1(x))$ is larger than ~\eqref{eq:suff1} and \eqref{eq:suff2} to ensure the MDS property for $\mathcal{C}_1(k,r,d,p)$ and $\mathcal{C}_2(k,r,d,p)$ respectively. Although this lower bound given in Theorem~\ref{thm:7}  is exponentially increasing on $k,d$, and $r$, we will show that the lower bound can be greatly reduced when the parameters are specified in the following. 

If 2 is a primitive element in $\mathbb{Z}_p$ and $\tau$ is a power of $p$, $h(x)$ is irreducible and $\mathbb{F}_2[x]/(h(x))$ is the finite field $\mathbb{F}_{2^{(p-1)\tau}}$~\cite{Itoh1991Characterization}. If $d-k+1=p$ for $\mathcal{C}_1(k,r,d,p)$ or $\frac{r}{2}=p$ for $\mathcal{C}_2(k,r,d,p)$, then $\tau$ is a power of $p$.  According to Theorem \ref{thm:mds1}, the MDS condition is reduced to that the determinant of each sub-matrix is non-zero in $\mathbb{F}_2[x]/(h(x))$. It can be shown, by computer search, that $\mathcal{C}_1(k,5,k+2,3)$ is an MDS code for $k=3,4,\ldots,12$. 
If $\tau$ is a power of 2, we have
\begin{equation}
h(x)=1+x^{\tau}+\cdots+x^{(p-1)\tau}=(1+x+\cdots+x^{p-1})^{\tau}.
\label{eq:check-factor}
\end{equation}
As $1+x+\cdots+x^{p-1}$ is irreducible \cite{Fenn1997Bit}, we can directly have the next theorem by Theorem \ref{thm:mds1}.

\begin{theorem}
If $\tau$ is a power of 2, 
and the determinant of any $\ell\times \ell$ sub-matrix of \eqref{matrixp} for $\ell=1,2,\ldots,\min\{k,r\}$ (or any $r\times r$ sub-matrix of \eqref{matrixh}) is invertible in $\mathbb{F}_2[x]/(1+x+\cdots+x^{p-1})$, then $\mathcal{C}_1(k,r,d,p)$ (or $\mathcal{C}_2(k,r,d,p)$) satisfy the MDS property.
\label{thm:MDSs}
\end{theorem}

Note that, by Theorem \ref{thm:MDSs}, the codes given in Example~\ref{example1} and Example~\ref{example2} are not MDS codes due to the fact that there exist  sub-matrices for both codes whose determinants are divisible by $1+x+x^2$.

Next, we characterize the detailed MDS condition for codes with some specific parameters.
The MDS condition of $\mathcal{C}_1(k,3,d,p)$ given in~\cite{hanxu2017triple} is a special case and is summarized in the next theorem. 
\begin{theorem}~\cite[Theorem 2]{hanxu2017triple}
Let $k\geq 4$. If $p\ge 2k-1$ is a prime such that 2 is a primitive element in $\mathbb{Z}_p$,  then the code $\mathcal{C}_1(k,3,d,p)$ satisfies the MDS property.
\label{thm:MDS5}
\end{theorem}


Next, we consider the MDS condition of $\mathcal{C}_2(k,4,d,p)$. When $r=4$, we have $\tau=2^{k}$. By Theorem~\ref{thm:MDSs}, we need to prove that the determinant of each $4 \times 4$ sub-matrix of $\mathbf{H}_{4\times (k+4)}$ \eqref{check4} is not divisible by $1+x+\cdots+x^{p-1}$. For any two positive integers $i,j$ such that $i<j$, we have
\[
x^i+x^j=x^i(1+x^{j-i})=x^i(1+x)(1+x+\cdots+x^{j-i-1}).
\]
A polynomial with even number of terms can be written as multiple pairs of $x^i+x^j$, we thus have that a polynomial with even number of terms must have a factor of $1+x$. It is easy to check that any determinant $f(x)$ of the $4 \times 4$ sub-matrix has even number of terms and $f(x)$ can be written as $f(x)=(1+x)g(x)$. Suppose that the determinant $f(x)$ is not divisible by $1+x+\cdots+x^{p-1}$, then $g(x)$ is not divisible by $1+x+\cdots+x^{p-1}$ because $1+x+\cdots+x^{p-1}$ is irreducible polynomial and is not a multiple of $1+x$. Recall that the polynomial $1+x^p$ can be factorized as $1+x^p=(1+x)(1+x+\cdots+x^{p-1})$, so $f(x)$ is not divisible by $1+x^p$. Conversely, if $f(x)$ is not divisible by $1+x^p$, then we can directly have that $f(x)$ is not divisible by $1+x+\cdots+x^{p-1}$. Therefore, it is sufficient to show that the determinant is not divisible by $1+x^p$.
An upper bound of $p$ for which $\mathcal{C}_2(k,4,d,p)$ is MDS is summarized in the next theorem.

\begin{theorem}
Let $p$ be a prime such that 2 is a primitive element in $\mathbb{Z}_p$. If $p$ is larger than 
\begin{equation}
(k-1)\cdot 2^k+17,
\label{eq:valuep}
\end{equation}
then $\mathcal{C}_2(k,4,d,p)$ is MDS for $k\geq 6$.
\label{thm:MDS2}
\end{theorem}
\begin{proof}
As $r=4$, we have $\frac{r}{2}=2$ that is strictly less than the value in \eqref{eq:valuep}. 
Note that each entry in \eqref{check4} is a polynomial with at most one non-zero term. For each $4\times 4$ sub-matrix of \eqref{check4}, each entry is also a polynomial with at most one non-zero term. By expanding the determinant of $4\times 4$ sub-matrix,  the determinant is a polynomial over $\mathbb{F}_2[x]$ with at most 24 non-zero terms  and can be written as
\[
x^{e_{1}}+x^{e_{2}}+\cdots+x^{e_{t}}=x^{e_{1}}(1+x^{e_{2}-e_{1}}+\cdots+x^{e_{t}-e_{1}}),
\]
where $t$ is a positive even number with $t\leq 24$, and $e_{1}<\cdots <e_{t}$.
If $e_{t}-e_{1}<p-1$, then $e_i-e_1<p-1$ for $i=2,\ldots,t$ and the determinant is a non-zero polynomial in $\mathbb{F}_2[x]/(1+x+\cdots+x^{p-1})$. If $e_{1+t/2}<p-1$, then $e_i<p-1$ for $i=1,2,\ldots,1+t/2$ and there exists at least one $i$ such that $e_i$ and $e_{j}$ are not congruent modulo $p$ for $j=1,\ldots,i,i+1,\ldots,t$. This is due to the fact that the number of $e_i$ that can be chosen from is larger than the number of $e_j$ for $j=2+t/2, \ldots, t$, and $e_j<p-1$ for $j=1,2,\ldots, i,i+1,\ldots, 1+t/2$. Hence, the determinant is not divisible by $1+x^{p}$. Thus,  $e_{t}-e_{1}\ge p-1$, $e_{1+t/2}\ge p-1$ is the remaining case to be proved.

Let $\mathbf{H}_1$ and $\mathbf{H}_2$  be the sub-matrices of $\mathbf{H}_{4\times (k+4)}$ given in \eqref{check4} with the column-vector indices being $1,2,3,k+2,k+3,k+4$ and 4 to $k+1$ respectively. Hence, we have
\begin{equation}
\mathbf{H}_1 \triangleq
\begin{bmatrix}
 1&x&0 & 0 \\
 1 & x^{2} & 0 & 0\\
 1 & x^{4} & 1 & 1\\
1 & 1 & x^{4} & x^{2\cdot 2^{k}}\\
  0 & 0 & x^{2} & x^{2^{k}} \\
 0 & 0  &x&1
 \end{bmatrix}^T
\label{matrixh4}
\end{equation}
and
\begin{equation}
\mathbf{H}_2 \triangleq
\begin{bmatrix}
 1 &x^8& x^{2^{k}}& x^{k\cdot 2^{k}}\\
 1& x^{16}&x^{2^{k-1}} & x^{(k-1)\cdot 2^{k}}\\
 \vdots& \vdots&\vdots&  \vdots  \\
 1 & x^{2^{k}} &x^{8} & x^{3\cdot 2^{k}}
 \end{bmatrix}^T.
\label{matrixh5}
\end{equation}

 For $k\geq 6$, it is sufficient to show that the determinant of the matrix consisting of any $g$ rows of \eqref{matrixh4} and any $4-g$ rows of \eqref{matrixh5} is not divisible by $1+x^{p}$ for $g=0,1,2,3,4$.

When $g=0$, the $4\times 4$ sub-matrix with column-vector indices being $i,j,\ell,m$ is 
\[
\begin{bmatrix}
 1 &x^{2^{i+2}}&  x^{2^{k-i+1}} & x^{(k-i+1)\cdot 2^{k}}\\
1 & x^{2^{j+2}}&x^{2^{k-j+1}} &x^{(k-j+1)\cdot 2^{k}} \\
1& x^{2^{\ell+2}}&x^{2^{k-\ell+1}}& x^{(k-\ell+1)\cdot 2^{k}}   \\
 1& x^{2^{m+2}} & x^{2^{k-m+1}} & x^{(k-m+1)\cdot 2^{k}}
 \end{bmatrix}^T,
\]
where $1\leq i<j<\ell<m\leq k-2$. The determinant of it is
\begin{equation}
\sum_{\ell_1\neq\ell_2\neq\ell_3\in \{i,j,\ell,m\}}x^{2^{\ell_1+2}+2^{k-\ell_2+1}+(k-\ell_3+1)2^{k}},
\label{deterC}
\end{equation}
where the number of non-zero term is $t$. Since the maximum degree and the minimum degree of terms in~\eqref{deterC} is  respectively $2^{m+2}+2^{k-j+1}+(k-i+1)2^k$ and $2^{i+2}+2^{k-\ell+1}+(k-m+1)2^k$, 
\[
e_{t}-e_{1}=(m-i)2^k+2^{m+2}+2^{k-j+1}+2^{i+2}+2^{k-\ell+1}.
\]
Clearly, $e_{t}-e_{1}$ is maximal when $m=k-2,\ell=k-3,i=1,j=2$.  Therefore,
$e_{t}-e_{1}\leq (k-2)2^k+2^{k-1}-24$ and,  according to \eqref{eq:valuep},
it is less than $p-1$. 

Now consider $g=1$. When the $4\times 4$ sub-matrix is consisted of the fourth column-vector of $\mathbf{H}_1$ in \eqref{matrixh4} and column-vectors 1, 2, 3 of $\mathbf{H}_2$ in \eqref{matrixh5}, i.e., 
\[
\begin{bmatrix}
1 & 1 & x^4 & x^{2\cdot 2^k} \\
1 & x^8 & x^{2^k} & x^{k\cdot 2^k} \\
1 & x^{16} & x^{2^{k-1}} & x^{(k-1)\cdot 2^k} \\
1 & x^{32} & x^{2^{k-2}} & x^{(k-2)\cdot 2^k} 
\end{bmatrix}^T,
\]
the determinant of the above matrix is
\begin{align*}
\small
& x^{2\cdot 2^k+2^{k-2}+8}+ x^{2\cdot 2^k+2^{k-2}+16}+ x^{2\cdot 2^k+ 2^{k-1}+8}+ x^{2\cdot 2^k+ 2^{k-1}+32}\\
&+ x^{2\cdot 2^k+ 2^{k}+16}+ x^{2\cdot 2^k+ 2^{k}+32}+x^{(k-2)\cdot 2^k+12}+x^{(k-2)\cdot 2^k+20}\\
&  +x^{(k-2)\cdot 2^k+ 2^{k-1}}+x^{(k-2)\cdot 2^k+ 2^{k-1}+8}+x^{(k-2)\cdot 2^k+2^{k}}+  \\
& x^{(k-1)\cdot 2^k+12}+x^{(k-2)\cdot 2^k+ 2^{k}+16}+ x^{(k-1)\cdot 2^k+36}+ \\
&x^{(k-1)\cdot 2^k+ 2^{k-2}}+ x^{(k-1)\cdot 2^k+ 2^{k-2}+8}+ x^{(k-1)\cdot 2^k+ 2^{k}}+  \\
& x^{k\cdot 2^k+20}+x^{(k-1)\cdot 2^k+ 2^{k}+32}+ x^{k\cdot 2^k+36}+ x^{k\cdot 2^k+ 2^{k-2}}+ \\
&x^{k\cdot 2^k+ 2^{k-2}+16}+ x^{k\cdot 2^k+ 2^{k-1}}+ x^{k\cdot 2^k+ 2^{k-1}+32}.
\end{align*}
If $k\geq 8$, then the above polynomial has 24 terms and the exponents in the polynomial are in ascending order. Note that $e_{13}=(k-1)\cdot 2^k+16$, which is less than $p-1$. If $k=6$, the above polynomial becomes
\begin{align*}
& x^{152}+ x^{160}+ x^{168}+ x^{192}+ x^{208}+ x^{224}+ x^{268}+ x^{276}+ \\
&x^{288}+x^{296}+ x^{320}+x^{332}+x^{344}+x^{356}+x^{384}+x^{400}+\\
&x^{404}+x^{416}+x^{420}+x^{448}.
\end{align*}
The lower bound of $p-1$ in \eqref{eq:valuep} is 337 when $k=6$. Hence,  $e_{11}=320$ in the above equation  is less than $p-1$.  When $k=7$, the determinant becomes
\begin{align*}
& x^{296}+x^{304}+x^{328}+ x^{352}+x^{400}+ x^{416}+x^{652}+x^{660}+ \\
&x^{704}+ x^{712}+ x^{768}+x^{780}+x^{784}+x^{800}+x^{804}+x^{808}+\\
&x^{896}+x^{916}+x^{932}+x^{944}+x^{960}+x^{992}.
\end{align*}
The lower bound of $p-1$ in \eqref{eq:valuep} is 785 when $k=7$. Therefore, $e_{12}=780$ is less than $p-1$.
After expanding the determinant for all the other $4\times 4$ sub-matrices when $g=1$, we can also determine that $e_{1+t/2}<p-1$.

When $g=2$, we first consider the $4\times 4$ sub-matrix  consisting of the third column-vector, the fourth column-vector of \eqref{matrixh4} and column-vectors 1, 2 of \eqref{matrixh5}, i.e., 
\[
\begin{bmatrix}
1 & x^4 & 1 & 1 \\
1 & 1 & x^{4} &x^{2\cdot 2^k}  \\
1 & x^{8} & x^{2^{k}} & x^{k\cdot 2^k} \\
1 & x^{16} & x^{2^{k-1}} & x^{(k-1)\cdot 2^k} 
\end{bmatrix}^T.
\]
The determinant of this sub-matrix is 
\begin{align*}
& x^{12}+ x^{20}+ x^{2^{k-1}}+ x^{2^{k-1}+8}+ x^{2^{k}+4}+x^{2^{k}+16}+x^{2 \cdot 2^{k}+8}+\\
& x^{2 \cdot 2^{k}+16}+ x^{2 \cdot 2^{k}+2^{k-1}+4}+ x^{2 \cdot 2^{k}+2^{k-1}+8}+x^{2 \cdot 2^{k}+2^{k}+4}+\\
&  x^{2 \cdot 2^{k}+2^{k}+16}+ x^{(k-1) \cdot 2^{k}}+ x^{(k-1) \cdot 2^{k}+12}+ x^{(k-1) \cdot 2^{k}+2^{k}}+ \\
& x^{k \cdot 2^{k}+8}+ x^{k \cdot 2^{k}+16}+ x^{k \cdot 2^{k}+20}+ x^{k \cdot 2^{k}+2^{k-1}}+ x^{k \cdot 2^{k}+2^{k-1}+4}.
\end{align*}
When $k\geq 6$, we have $t=20$ and $e_{11}=3\cdot 2^k+4$ is less than $p-1$. Similarly, we can prove $e_{1+t/2}<p-1$ for the other cases when $g=2$.
When $g=3,4$, it is easy to check that either $e_{t}-e_{1}$ or $e_{1+t/2}$ is less than $p-1$. 

From the above discussion, the determinants of all $4 \times 4$ sub-matrices in \eqref{check4} are not divisible by $1+x^{p}$. The codes $\mathcal{C}_2(k,4,d,p)$ thus satisfy MDS property for $k\geq 6$.
\end{proof}
\begin{table*}[tbh]
\caption{All values of $p$ that $\mathcal{C}_2(k,4,d,p)$ are MDS codes for $r=4$ and  $k=2,3,\ldots,13$.}
\begin{center}
\begin{tabular}{|c|c|c|c|c|c|c|} \hline
$k$&2& 3 & 4 & 5 & 6 & 7 \\ \hline 
$p$&$\geq 11$&11, $\geq 19$ & 19, $\geq 37$&$\geq 19,$  $\neq 29,61$&$\geq 19,$ $\neq 29,37,61,107$&$\geq 53$, $\neq 61,107$  \\ \hline \hline
$k$ & 8 & 9 & 10 & 11 & 12 & 13\\ \hline 
$p$&$\geq 53,$ & $\geq 53,$ & $\geq 67,$ $\neq 107,$ &67,101,131,149, &67,101,131,149,$\geq 179$ &67,101,131, $\geq 179,$ $\neq 211,$\\ 
& $\neq 61,$ &  $\neq 59,61,$& $139,163,491,$ &$\geq 173,$ $\neq 491,$&  $\neq 491,509,613,$& $347,491,509,613,653,709,$\\
& $107,163,$ &  $107,139,$& $509,613,$ &  $509,613,709,$&$653,709,1741,2027,$ &$1741,1949,1973,2027,4093,$\\ 
& $491,$&$163,491,$ & $709,1741,$ & $1741,2027,$ & $4093,6827,8171,$ &$6827,8171,16363,16381,$ \\ 
& $509$ &$509$ & $2027$ & $4093,6827$ & $36353,39937$ & $36353,39937,80897$ \\ \hline
\end{tabular}
\end{center}
\label{table:p-example}
\end{table*}
Indeed, the upper bound of $p$ in Theorem \ref{thm:MDS2} is exponential in $k$. However, since we are  interested in small $k$, we can first compute each $4\times 4$ determinant that can be viewed as a polynomial $g(x)$ in $\mathbb{F}_2[x]$. Then we can  check by computer search whether polynomial $g(x)$  is a multiple of $1+x^p$ or not. By using this procedure,  all values of $p$ for which the codes $\mathcal{C}_2(k,4,d,p)$ are  MDS codes are found and summarized in Table \ref{table:p-example} for $k=2,3,\ldots,13$. 

\section{Weak-Optimal Repair Procedure for One Column Failure}
\label{sec:repair}

In this section, we demonstrate how to recover the bits stored in any information column for the first construction when an information column is failed. We also present an efficient repair algorithm for the second construction when a single column is failed. We prove that the proposed procedures are  with asymptotically weak-optimal repair bandwidth.
\subsection{Repair Procedure of the First Construction}
\label{sec:repair1}
In this subsection, we  assume that the information column $f$ is erased, where $1\le f\le k$. We want to recover the bits $s_{0,f}$, $s_{1,f}, \ldots, s_{(p-1)\tau-1,f}$ stored in the information column $f$ by accessing bits from $k-1$ other information columns and $d-k+1$ parity columns.
Recall that we can compute the extra bits by \eqref{eq:check1}. For notational convenience, we refer the \emph{bits} of column $i$ as the $p\tau$ bits $s_{0,i}$, $s_{1,i}, \ldots, s_{p\tau-1,i}$ in this section. 

For $j=1,2,\ldots,r$ and $0\leq \ell \leq p\tau-1$, we define the $\ell$-th parity set of the $j$-th parity column as 
$$P_{\ell,j}=\{s_{\ell-(j-1)\eta^{0},1},s_{\ell-(j-1)\eta,2},\ldots,s_{\ell-(j-1)\eta^{k-2},k-1},s_{\ell,k}\},$$
for $1\leq j\leq d-k+1$, and 
$$P_{\ell,j}=\{s_{\ell,1},s_{\ell-(2\eta-j)\eta^{k-2},2},\ldots,s_{\ell-(2\eta-j)\eta,k-1},s_{\ell-(2\eta-j),k}\},$$
for $d-k+2\leq j\leq r$. Note that all the indices of the elements in $P_{\ell,j}$  are taken modulo $p\tau$. It is clear that parity set $P_{\ell,j}$ consists of information bits which are used to generate the parity bit $s_{\ell,k+j}$. That is,
\begin{equation*}
s_{\ell,k+j}=\left\{\begin{array}{ll}&s_{\ell,k}+\sum_{i=1}^{k-1}s_{\ell-(j-1)\eta^{i-1},i}\\
&\text{ for } 0\le \ell\le p\tau-1, 1\leq j\leq d-k+1;\\
&s_{\ell,1}+\sum_{i=2}^{k}s_{\ell-(2\eta-j)\eta^{k-i},i} \\
&\text{ for } 0\le \ell\le p\tau-1, d-k+2\leq j\leq r.
\end{array}\right.
\end{equation*}
When we say an information bit is repaired by a parity column, it means that we access the parity bits of the parity column, and all the information bits, excluding the erased bits, in this parity set. Assume that the information column $f$ has failed. When $1\leq j\leq d-k+1$, we can repair the $\ell$-th bit in this failed column by the $j$-th parity column:
\begin{equation}
\label{eq:recover1}
s_{\ell,f}=\left\{\begin{array}{ll}&s_{\ell+(j-1)\eta^{f-1},k+j}+s_{\ell+(j-1)\eta^{f-1},k}+\\
&\sum_{i=1,i\neq f}^{k-1}s_{\ell+(j-1)\eta^{f-1}-(j-1)\eta^{i-1},i}, 1\leq f\leq k-1;\\
&s_{\ell,k+j}+\sum_{i=1}^{k-1}s_{\ell-(j-1)\eta^{i-1},i},  f=k.
\end{array}\right.
\end{equation}
When $d-k+2\leq j \leq r$, we can repair the bit $s_{\ell,f}$ by the $j$-th parity column:
\begin{equation}
\label{eq:recover2}
s_{\ell,f}=\left\{\begin{array}{ll}&s_{\ell,k+j}+\sum_{i=2}^{k}s_{\ell-(2\eta-j)\eta^{k-i},i},  f=1;\\
&s_{\ell+(2\eta-j)\eta^{k-f},k+j}+s_{\ell+(2\eta-j)\eta^{k-f},1}+\\
&\sum_{i=2,i\neq f}^{k}s_{\ell+(2\eta-j)\eta^{k-f}-(2\eta-j)\eta^{k-i},i}, 2\leq f\leq k.
\end{array}\right.
\end{equation}



\begin{algorithm}
\begin{algorithmic}[1]
\STATE {Suppose that the information column $f$ has failed.}
    \IF {$f\in\{1, 2,\ldots, \lceil k/2 \rceil\}$}
            \FOR {$\ell \bmod \eta^{f} \in\{0,1,\ldots,\eta^{f-1}-1\}$}
                  \STATE {Repair the bit $s_{\ell,f}$ by the first parity column, i.e., by \eqref{eq:recover1} with $j=1$.}
            \ENDFOR
         \FOR {$t=1,2,\ldots,d-k$}
            \FOR {$\ell \bmod \eta^f \in\{t\eta^{f-1},t\eta^{f-1}+1,\ldots,(t+1)\eta^{f-1}-1\}$}
                  \STATE {Repair the bit $s_{\ell,f}$ by parity column $d-k-t+2$, i.e., by \eqref{eq:recover1} with $j=d-k-t+2$.}
            \ENDFOR
         \ENDFOR
         \RETURN{}
    \ENDIF
     \IF {$f\in\{\lceil k/2 \rceil+1, \lceil k/2 \rceil+2,\ldots, k\}$}
            \FOR {$\ell \bmod \eta^{k+1-f} \in\{0,1,\ldots,\eta^{k-f}-1\}$}
                  \STATE {Repair the bit $s_{\ell,f}$ by the first parity column, i.e., by \eqref{eq:recover1} with $j=1$.}
            \ENDFOR
            \FOR {$t=1,2,\ldots,d-k$}
                 \FOR {$\ell \bmod \eta^{k+1-f} \in\{t\eta^{k-f},t\eta^{k-f}+1,\ldots,(t+1)\eta^{k-f}-1\}$}
                     \STATE {Repair the bit $s_{\ell,f}$ by parity column $d-k+t+1$, i.e., by \eqref{eq:recover2} with $j=d-k+t+1$.}
            \ENDFOR
         \ENDFOR
    \RETURN{}     
    \ENDIF
    \caption{Repair procedure of the first construction for one information failure}
    \label{alg:A1}
\end{algorithmic}
\end{algorithm}

The repair algorithm is stated in Algorithm \ref{alg:A1}. In the algorithm, we divide the $k$ information columns into two parts. The first part has $\lceil k/2 \rceil$ columns and the second part has $k-\lceil k/2 \rceil$ columns. If a column in the first part fails, we repair the failure column by the first $d-k+1$ parity columns; otherwise,  the failure column is recovered by the first parity column and the last $d-k$ parity columns.

Let's again consider the Example~\ref{example1} to illustrate the repair procedure in detail. In this example, $k=4$, $d=5$, $p=3$, $r=3$, $\tau=4$, and $\eta=2$. 
The repair procedure of the first column is shown in Fig.~\ref{example} followed from Algorithm \ref{alg:A1} with $f=1$.
Suppose that the second information column (i.e., node 2) fails, i.e,. $f=2$. By steps 3 to 4 in Algorithm~\ref{alg:A1}, we can repair the bits $s_{0,2},s_{1,2},s_{4,2},s_{5,2}$ by
\begin{eqnarray*}
& s_{0,2}=s_{0,1}+s_{0,3}+s_{0,4}+(s_{0,1}+s_{0,2}+s_{0,3}+s_{0,4}) \\
& s_{1,2}=s_{1,1}+s_{1,3}+s_{1,4}+(s_{1,1}+s_{1,2}+s_{1,3}+s_{1,4}) \\
& s_{4,2}=s_{4,1}+s_{4,3}+s_{4,4}+(s_{4,1}+s_{4,2}+s_{4,3}+s_{4,4}) \\
& s_{5,2}=s_{5,1}+s_{5,3}+s_{5,4}+(s_{5,1}+s_{5,2}+s_{5,3}+s_{5,4}).
\end{eqnarray*}
Similarly, we can repair the bits $s_{2,2},s_{3,2},s_{6,2},s_{7,2}$ by
\begin{align*}
s_{2,2}&=s_{3,1}+s_{0,3}+s_{4,4}+(s_{3,1}+s_{2,2}+s_{0,3}+s_{4,4}) \\
s_{3,2}&=s_{4,1}+s_{1,3}+s_{5,4}+(s_{4,1}+s_{3,2}+s_{1,3}+s_{5,4}) \\
s_{6,2}&=s_{7,1}+s_{4,3}+s_{0,4}+s_{4,4}+
(\mathbf{s_{11,1}}+\mathbf{s_{10,2}}+\\
&\mathbf{s_{8,3}}+s_{0,4})+(s_{3,1}+s_{2,2}+s_{0,3}+s_{4,4}) \\
s_{7,2}&=s_{0,1}+s_{4,1}+s_{5,3}+s_{1,4}+s_{5,4}+(s_{0,1}+\mathbf{s_{11,2}}\\
&+\mathbf{s_{9,3}}+s_{1,4})+(s_{4,1}+s_{3,2}+s_{1,3}+s_{5,4}).
\end{align*}
As a result, the 8 bits stored in the second information column can be recovered by downloading 6 bits from the first information column and 4 bits from each of columns 3 to 6. There are total 22 bits downloaded in the repair procedure. It can be verified that for the code in Example~\ref{example1}, column 3 and column 4 can be rebuilt by accessing 22 bits and 20 bits from 5 columns, respectively.

There exist some intersections of information bits between different parity sets. The key idea in Algorithm \ref{alg:A1} is that for each erased information column, the accessed parity sets are carefully chosen such that they have a larger intersection. This  selection results in a small number of accesses during the repair process. Moreover, it is clear that the choice of the encoding vectors is crucial if we want to ensure the MDS property and efficient repair. We show in the next theorem that Algorithm \ref{alg:A1} can recover any information column and the repair bandwidth of one information column is asymptotically weak-optimal.
\begin{theorem}
Assume that column $f$ is failed. All the information bits in column $f$ can be recovered by Algorithm \ref{alg:A1}, where $1\leq f\leq k$. The repair bandwidth of information column $f$ induced by Algorithm \ref{alg:A1} is
\begin{equation}
(p-1)((d+1)\eta^{k-3}-\eta^{k-f-2}),
\label{eq:repairband}
\end{equation}
when $f\in\{1,2,\ldots,\lceil k/2\rceil \}$, and is 
\begin{equation}
(p-1)((d+1)\eta^{k-3}-\eta^{f-3}),
\label{eq:repairband1}
\end{equation}
when $f\in\{\lceil k/2\rceil+1,\lceil k/2\rceil+2,\ldots,k\}$.
\label{thmrep}
\end{theorem}
\begin{proof}
Recall that $\eta=d-k+1$ and $\tau=\eta^{k-2}$. We can recover all $(p-1)\eta^{k-2}$ information bits in column $f$ by \eqref{eq:recover1} or \eqref{eq:recover2} with a given $j$, i.e., downloading the other $k-1$ information columns and one extra $j$-th parity column, where $j=1,2,\ldots,r$. The total  repair bandwidth by this repair procedure is $k(p-1)\eta^{k-2}$. In order to  reduce the repair bandwidth, in additional to the $k-1$ healthy information columns, $d-k+1$ parity columns are used to recover column $f$ in Algorithm~\ref{alg:A1}. 
$(p-1)\eta^{k-2}$ information bits in column $f$ are recovered by $d-k+1$ parity columns, where each parity column recovers $(p-1)\eta^{k-3}$ bits. 
When $1\leq f\leq \lceil k/2\rceil$,  the number of bits in each column is a multiple of $\eta^f$ when $k\geq 4$. The values of $\ell\bmod \eta^f$ are partitioned into $\eta$ groups each with size $\eta^{f-1}$,  where each group is associated with one parity column. $(p-1)\eta^{k-3}$ bits to be recovered are then associated with a group according to the values of their indices divided by $\eta^f$. The idea behind this grouping is that we can  employ the cyclic structure in the underlying quotient ring  to recover some specific information by choosing a suitable parity column such that  the number of bits downloaded is minimized.
 Only small portion of information bits in each  column are downloaded and a particular subset of the $d-k+1$ parity columns are chosen to have many common bits.  
 When $f> \lceil k/2\rceil$, we can replace $\ell\bmod \eta^f$ with $\ell\bmod \eta^{k+1-f}$ and apply the similar procedure. 
 In the following, we show that both  column $f$ can be repaired and the repair bandwidth can be reduced by Algorithm~\ref{alg:A1}.

Assume that  $1\leq f\leq \lceil k/2\rceil$.  By steps 3 and 4 in Algorithm \ref{alg:A1}, the bits $s_{\ell,f}$ are recovered by the first parity column, i.e., by \eqref{eq:recover1} with $j=1$ when 
\begin{equation}
\ell\bmod \eta^f \in \{0,1,\ldots,\eta^{f-1}-1\}
\label{eq:p1}
\end{equation}
and $\ell\in\{0,1,\ldots,(p-1)\eta^{k-2}-1\}$. As $\ell$ ranges from $0$ to $(p-1)\eta^{k-2}-1$ and $(p-1)\eta^{k-2}$ is a multiple of $\eta^f$ for $k\geq 4$, $\ell \bmod \eta^f$ is uniform distributed over $\{0, 1,\ldots, \eta^f-1\}$. We thus obtain that the total number of bits $s_{\ell,f}$ that are recovered by \eqref{eq:recover1} with $j=1$ is $\frac{(p-1)\eta^{k-2}}{\eta^f}\cdot \eta^{f-1}=(p-1)\eta^{k-3}$. 

We have recovered $(p-1)\eta^{k-3}$ bits $s_{\ell,f}$ in column $f$ with indices in~\eqref{eq:p1} by the first parity column. As column $f$ contains $(p-1)\eta^{k-2}$ information bits, we still need to repair the other $(\eta-1)(p-1)\eta^{k-3}$ information bits in column $f$. For simplicity of notation, let $\delta=d-k-t+1$. By steps 5 to 7 in Algorithm \ref{alg:A1}, for $t=1,2,\ldots,d-k$,
the bits $s_{\ell,f}$ are recovered by parity column $\delta+1$, i.e., by \eqref{eq:recover1} with $j=\delta+1$ when  
\begin{equation}
\ell\bmod \eta^f \in \{t\eta^{f-1},t\eta^{f-1}+1,\ldots,(t+1)\eta^{f-1}-1\}
\label{eq:p2}
\end{equation}  
and $\ell\in\{0,1,\ldots,(p-1)\eta^{k-2}-1\}$. As $\ell \bmod \eta^f$ is uniform distribution, the number of bits $s_{\ell,f}$ that are recovered by \eqref{eq:recover1} with $j=\delta+1$  is $(p-1)\eta^{k-3}$. Therefore, the total number of bits $s_{\ell,f}$ that are recovered by parity column $\delta+1$ for $t=1,2,\ldots,\eta-1$ is $(\eta-1)(p-1)\eta^{k-3}$.

We have recovered $(p-1)\eta^{k-3}+(\eta-1)(p-1)\eta^{k-3}=(p-1)\eta^{k-2}$ bits $s_{\ell,f}$ with indices in~\eqref{eq:p1} and \eqref{eq:p2} by Algorithm~\ref{alg:A1}. Since
\begin{align*}
&\{t\eta^{f-1},t\eta^{f-1}+1,\ldots,(t+1)\eta^{f-1}-1\}\cap\\
&\{t'\eta^{f-1},t'\eta^{f-1}+1,\ldots,(t'+1)\eta^{f-1}-1\}=\emptyset 
\end{align*}
for $0\le t\neq t'\le \eta-1$,
 the indices of all the repaired bits in column $f$ are distinct and all $(p-1)\eta^{k-2}$ information bits are recovered by Algorithm \ref{alg:A1} for $1\leq f\leq \lceil k/2\rceil$. Similarly, when $\lceil k/2\rceil+1\leq f\leq k$, we can also show that we can recover all the information bits in column $f$ by Algorithm \ref{alg:A1}.


Next, we calculate the repair bandwidth of column $f$ by Algorithm \ref{alg:A1} when $1\leq f\leq \lceil k/2\rceil$. Recall that by steps 3 and 4 in Algorithm \ref{alg:A1}, $(p-1)\eta^{k-3}$ bits $s_{\ell,f}$ with indices in \eqref{eq:p1} are recovered by \eqref{eq:recover1} with $j=1$, where $(p-1)\eta^{k-3}$ parity bits $s_{\ell,k+1}$ and $(k-1)(p-1)\eta^{k-3}$ information bits $s_{\ell,i}$ with $i=1,2,\ldots,f-1,f+1,\ldots, k$ are needed to perfrom the repair procedure. By steps 5 to 7 in Algorithm \ref{alg:A1}, $(\eta-1)(p-1)\eta^{k-3}$ bits $s_{\ell,f}$ with indices in \eqref{eq:p2} are recovered by \eqref{eq:recover1} with $j=\delta+1$, where $(\eta-1)(p-1)\eta^{k-3}$ parity bits $s_{\ell,k+\delta+1}$, $(\eta-1)(k-1)(p-1)\eta^{k-3}$ information bits $s_{\ell+\delta\eta^{f-1}-\delta\eta^{i-1},i}$ with $i=1,2,\ldots,f-1,f+1,\ldots, k-1$ and $s_{\ell+\delta\eta^{f-1},k}$ are needed with $\delta=1,2,\ldots,\eta-1$. Note that there are many needed bits that are in common such that the total number of downloaded bits is less.

We first download $(k-1)(p-1)\eta^{k-3}$ information bits $s_{\ell,i}$ for $i=1,2,\ldots,f-1,f+1,\ldots,k$ and $(p-1)\eta^{k-3}$ parity bits $s_{\ell,k+1}$ with indices in \eqref{eq:p1} in steps 3 and 4. Then, we only need to download the bits in steps 5 to 7 which have  not downloaded in steps 3 and 4. In the following, we show that we do not need to download all $(\eta-1)(k-1)(p-1)\eta^{k-3}$ information bits $s_{\ell+\delta\eta^{f-1}-\delta\eta^{i-1},i}$ with $i=1,2,\ldots,f-1,f+1,\ldots, k-1$ and $s_{\ell+\delta\eta^{f-1},k}$ in steps 5 to 7, as some of them have been downloaded in steps 3 and 4.

We now consider the needed information bits $s_{\ell+\delta\eta^{f-1}-\delta\eta^{i-1},i}$ with $i=1,2,\ldots,f-1,f+1,\ldots, k-1$, $\delta=1,2,\ldots,\eta-1$  in steps 5 to 7, where $\ell$ are in \eqref{eq:p2}. We want to show that many bits $s_{\ell+\delta\eta^{f-1}-\delta\eta^{i-1},i}$ are not necessary to be downloaded, as they are downloaded in steps 3 and 4. Given an $\ell$ that is in \eqref{eq:p2}, where $0\leq \ell \leq (p-1)\eta^{k-2}-1$, we have that the index of the corresponding needed bit $s_{\ell+\delta\eta^{f-1}-\delta\eta^{i-1},i}$ is $\ell'=(\ell+\delta\eta^{f-1}-\delta\eta^{i-1})\bmod p\eta^{k-2}$. Recall that the bits $s_{\ell,i}$ for $\ell\bmod \eta^f\in\{0,1,\ldots,\eta^{f-1}-1\}$ and $i=1,2,\ldots,f-1,f+1,\ldots,k-1$ are downloaded in steps 3 and 4. If 
\begin{align*}
&\ell'\bmod \eta^f=((\ell+\delta\eta^{f-1}-\delta\eta^{i-1})\bmod p\eta^{k-2})\bmod \eta^f\\
&\in\{0,1,\ldots,\eta^{f-1}-1\},
\end{align*}
then we do not need to download the bit $s_{\ell',i}$, as it is downloaded in steps 3 and 4. Otherwise, it should be downloaded. 

We first consider the bits $s_{\ell',i}$ for $i=1,2,\ldots,f-1$. If $\ell\bmod \eta^f=t\eta^{f-1}$, then there exists an integer $m$ such that $\ell=m\eta^f+t\eta^{f-1}$. Thus, we have
\begin{align*}
\ell'\bmod \eta^f=&((\ell+\delta\eta^{f-1}-\delta\eta^{i-1})\bmod p\eta^{k-2})\bmod \eta^f\\
=&(m\eta^f+t\eta^{f-1}+\delta\eta^{f-1}-\delta\eta^{i-1})\bmod \eta^f \\
=&((t+\delta)\eta^{f-1}-\delta\eta^{i-1})\bmod \eta^f\\
=&\eta^f-\delta\eta^{i-1}.
\end{align*}
If $\ell\bmod \eta^f=(t+1)\eta^{f-1}-1$, then there exists an integer $m$ such that $\ell=m\eta^f+(t+1)\eta^{f-1}-1$. Thus, we have
\begin{align*}
&\ell'\bmod \eta^f=((\ell+\delta\eta^{f-1}-\delta\eta^{i-1})\bmod p\eta^{k-2})\bmod \eta^f\\
=&(m\eta^f+(t+1)\eta^{f-1}-1+\delta\eta^{f-1}-\delta\eta^{i-1})\bmod \eta^f \\
=&((t+\delta+1)\eta^{f-1}-1-\delta\eta^{i-1})\bmod \eta^f\\
=&\eta^{f-1}-\delta\eta^{i-1}-1.
\end{align*}
By repeating the above procedure for $\ell\bmod\eta^f=t\eta^{f-1}+1,\ldots,(t+1)\eta^{f-1}-2$, we can obtain that 
\begin{eqnarray}
\ell'\bmod \eta^f =& \eta^{f}-\delta\eta^{i-1},\eta^{f}-\delta\eta^{i-1}+1,\ldots,\eta^{f}-1,0,\nonumber\\
&1,\ldots,\eta^{f-1}-\delta\eta^{i-1}-1
\label{eq:p3}
\end{eqnarray}
when $\ell\bmod \eta^f$ runs from $t\eta^{f-1}$ to $(t+1)\eta^{f-1}-1$.
Thus, the bits $s_{\ell',i}$ with $i=1,2,\ldots,f-1$ and $\ell'$ in the union set of all the values in \eqref{eq:p3} with $\delta=1,2,\ldots,\eta-1$ are needed. 
When $i=1,2,\ldots,f-1$, the union set of all the values in \eqref{eq:p3} with $\delta=1,2,\ldots,\eta-1$ is
\begin{eqnarray}
&&\cup_{\delta=1,\ldots,\eta-1}\{\eta^{f}-\delta\eta^{i-1},\ldots,\eta^{f}-1,0,\nonumber \\
&&1,\ldots,\eta^{f-1}-\delta\eta^{i-1}-1\}\nonumber\\
&=&\{\eta^{f}-(\eta-1)\eta^{i-1},\ldots,\eta^{f}-1,0,\nonumber \\
&&1,\ldots,\eta^{f-1}-\eta^{i-1}-1\}.
\label{eq:p4}
\end{eqnarray}
Since $i\le f-1$, we have $\eta^{f-1}-\eta^{i-1}-1<\eta^{f}-(\eta-1)\eta^{i-1}$. Hence, the elements in \eqref{eq:p4} can be rearranged as 
\begin{equation}
\label{set:range}
\{0,1,\ldots,\eta^{f-1}-\eta^{i-1}-1, \eta^{f}-(\eta-1)\eta^{i-1},\ldots,\eta^{f}-1\}.	
\end{equation}
Recall that $(f-1)(p-1)\eta^{k-3}$ information bits $s_{\ell,i}$ for $i=1,2,\ldots,f-1$ with $\ell\bmod \eta^f\in\{0,1,\ldots,\eta^{f-1}-1\}$ have already been downloaded in steps 3 and 4. Since $\eta^{f-1}-\eta^{i-1}-1\le \eta^{f-1}-1<\eta^{f}-(\eta-1)\eta^{i-1}$, 
\begin{align*}
&\{0,1,\ldots,\eta^{f-1}-\eta^{i-1}-1,\ldots, \eta^{f}-(\eta-1)\eta^{i-1},\ldots,\eta^{f}-1\} \\
&\setminus\{0,1,\ldots,\eta^{f-1}-1\}=\{\eta^{f}-(\eta-1)\eta^{i-1},\ldots,\eta^{f}-1\}.
\end{align*}
We thus only need to download $(\eta-1)(p-1)\eta^{k+i-f-3}$ information bits $s_{\ell',i}$ from columns $i$ for $i=1,2,\ldots,f-1$ with $\ell'\bmod \eta^f\in\{\eta^{f}-(\eta-1)\eta^{i-1},\ldots,\eta^{f}-1\}$ in steps 5 to 7.

We then consider the information bits $s_{\ell',i}$ for $i=f+1,f+2,\ldots,k-1$. Recall that $\ell'=(\ell+\delta\eta^{f-1}-\delta\eta^{i-1})\bmod p\eta^{k-2}$. Similarly, we can prove that 
\[
\ell'\bmod \eta^f =0,1,\ldots,\eta^{f-1}-1
\]
when $\ell \bmod \eta^f$ runs from $t\eta^{f-1}$ to $(t+1)\eta^{f-1}-1$.
In this case, all needed bits have already been downloaded in steps 3 and 4 and
we thus do not need to download bits from columns $i$ for $i=f+1,f+2,\ldots,k-1$ in steps 5 to 7.

In the last, we consider $i=k$, i.e., the needed bits $s_{\ell+\delta\eta^{f-1},k}$ with indices in \eqref{eq:p2}.
Given an $\ell$ that is in \eqref{eq:p2}, where $0\leq \ell \leq (p-1)\eta^{k-2}-1$, we have that the index of the corresponding needed bit $s_{\ell+\delta\eta^{f-1},k}$ is $\ell'=(\ell+\delta\eta^{f-1})\bmod p\eta^{k-2}$.
If $\ell\bmod \eta^f=t\eta^{f-1}$, then $\ell=m\eta^f+t\eta^{f-1}$, where $m$ is an integer. Thus, we have
\begin{align*}
\ell'\bmod \eta^f=&((\ell+\delta\eta^{f-1})\bmod p\eta^{k-2})\bmod \eta^f\\
=&(m\eta^f+t\eta^{f-1}+\delta\eta^{f-1})\bmod \eta^f\\
=&((t+\delta)\eta^{f-1})\bmod \eta^f\\
=&0.
\end{align*}
Hence, we can obtain that 
\[
\ell'\bmod \eta^f =0,1,\ldots,\eta^{f-1}-1
\]
when $\ell \bmod \eta^f$ runs from $t\eta^{f-1}$ to $(t+1)\eta^{f-1}-1$.
Recall that $(p-1)\eta^{k-3}$ information bits $s_{\ell,k}$ with indices in \eqref{eq:p1} have already been downloaded in steps 3 and 4, 
we thus do not need to download bits $s_{\ell,k}$ in steps 5 to 7.

We can count that the total number of bits downloaded from $d=k+(r-1)/2$ columns to repair the information column $f$ is
\begin{align*}
&\underbrace{k(p-1)\eta^{k-3}}_{\text{the first parity column}}+\\
&\underbrace{(d-k)(p-1)\eta^{k-3}+\sum_{i=1}^{f-1}(d-k)(p-1)\eta^{k+i-f-3}}_{\text{parity columns 2 to } d-k+1}\\
&=(p-1)((d+1)\eta^{k-3}-\eta^{k-f-2}),
\end{align*}
which is equal to \eqref{eq:repairband}. 

When $f> \lceil k/2 \rceil$, we can replace $\ell\bmod \eta^f$ with $\ell\bmod \eta^{k+1-f}$ and the repair bandwidth of column $f$ can be obtained by replacing $f$ in \eqref{eq:repairband} with $k+1-f$ via the similar argument. This completes the proof.
\end{proof}

When $f> \lceil k/2 \rceil$, the repair bandwidth of column $f$ is the same as that of column $k+1-f\le \lceil k/2 \rceil$ by \eqref{eq:repairband} and \eqref{eq:repairband1}. Therefore, we only consider the cases of $1\leq f\leq \lceil k/2 \rceil$ in the following discussion.
By Theorem~\ref{thmrep}, the repair bandwidth increases when $f$ increases.
When $f=1$, the repair bandwidth is $d(p-1)\eta^{k-3}$, which achieves the lower bound in \eqref{optimal_repair}. Even for the worst case of $f=\lceil k/2\rceil$, the repair bandwidth is
\begin{align*}
(p-1)((d+1)\eta^{k-3}-\eta^{k-\lceil k/2\rceil-2})< (p-1)(d+1)\eta^{k-3},
\end{align*}
which is strictly less than $\frac{d+1}{d}$ times of the value in \eqref{optimal_repair}. Therefore, the repair bandwidth of any one information failure can achieve the weak-optimal repair in \eqref{optimal_repair} asymptotically when $d$ is large enough.


\subsection{Repair Procedure of the Second Construction}
\label{sec:repair2}
Suppose that column $f$ is erased, where $1\leq f\leq k+r$. We want to recover the bits $s_{0,f},s_{1,f},\ldots,s_{(p-1)\tau-1,f}$ by accessing bits from $d=k+\frac{r}{2}-1$ healthy columns. We refer the \emph{bits} of column $i$ as the $p\tau$ bits $s_{0,i},s_{1,i},\ldots,s_{p\tau-1,i}$. 
Recall that the check matrix of the second construction is given in~\eqref{matrixh}. By the $j$-th row of the check matrix, where $j=1,2,\ldots,\frac{r}{2}$, we have
\begin{align*}
s_{\ell-(j-1),1}+s_{\ell-(j-1)\frac{r}{2},2}+\cdots+s_{\ell-(j-1)(\frac{r}{2})^{d-1},d}+s_{\ell,d+1}=0,
\end{align*}
for $\ell=0,1,\ldots, p\tau-1$. Thus, we can repair a bit $s_{\ell,f}$ by
\begin{equation}
s_{\ell,f}=\left\{\begin{array}{ll}&\sum_{i=1,i\neq f}^{d}s_{\ell+(j-1)(\frac{r}{2})^{f-1}-(j-1)(\frac{r}{2})^{i-1},i} \\
&+s_{\ell+(j-1)(\frac{r}{2})^{f-1},d+1}, 1\leq f\leq d;\\
&\sum_{i=1}^{d}s_{\ell-(j-1)(\frac{r}{2})^{i-1},i}, f=d+1.
\end{array}\right.
\label{eq:repair21}
\end{equation}
Similarly, by the $j$-th row for $j=\frac{r}{2}+1,\frac{r}{2}+2,\ldots, r-1$ of \eqref{matrixh}, we have
\begin{align*}
&s_{\ell,\frac{r}{2}+1}+s_{\ell-(r-j)(\frac{r}{2})^{d-1},\frac{r}{2}+2}+\cdots+s_{\ell-(r-j)(\frac{r}{2})^{1},n-1}\\
&+s_{\ell-(r-j),n}=0,
\label{repairl}
\end{align*}
for $\ell=0,1,\ldots, p\tau-1$. The bit $s_{\ell,f}$ in column $f$ for $\frac{r}{2}+1\leq f\leq n$ can be recovered by accessing $k+\frac{r}{2}-1$ bits
\begin{equation}
s_{\ell,f}=\left\{\begin{array}{ll}&\sum_{i=\frac{r}{2}+2,i\neq f}^{n}s_{\ell+(r-j)(\frac{r}{2})^{n-f}-(r-j)(\frac{r}{2})^{n-i},i}\\
&+s_{\ell+(r-j)(\frac{r}{2})^{n-f},\frac{r}{2}+1}, \frac{r}{2}+2 \leq f\leq n;\\
&\sum_{i=\frac{r}{2}+2}^{n}s_{\ell-(r-j)(\frac{r}{2})^{n-i},i}, f=\frac{r}{2}+1.
\end{array}\right.
\label{eq:repair22}
\end{equation}
When $j=r$, by the $r$-th row of \eqref{matrixh}, we can have another set of equations
\begin{align*}
&s_{\ell,\frac{r}{2}+1}+s_{\ell-(d-1)(\frac{r}{2})^{d-1},\frac{r}{2}+2}+s_{\ell-(d-2)(\frac{r}{2})^{d-1},\frac{r}{2}+3}+\\
&\cdots+s_{\ell-(\frac{r}{2})^{d-1},n-1}+s_{\ell,n}=0,
\end{align*}
for $\ell=0,1,\ldots, p\tau-1$. We can also recover the bit $s_{\ell,f}$ by
\begin{equation}
s_{\ell,f}=\left\{\begin{array}{ll}&\sum_{i=\frac{r}{2}+2,i\neq f}^{n}s_{\ell+(n-f)(\frac{r}{2})^{d-1}-(n-i)(\frac{r}{2})^{d-1},i}\\
&+s_{\ell+(n-f)(\frac{r}{2})^{d-1},\frac{r}{2}+1},  \frac{r}{2}+2 \leq f\leq n;\\
&\sum_{i=\frac{r}{2}+2}^{n}s_{\ell-(n-i)(\frac{r}{2})^{d-1},i}, f=\frac{r}{2}+1.
\end{array}\right.
\label{eq:repair23}
\end{equation}

Note that all the indices in this section and throughout the paper are taken modulo $p\tau$. 
For $j=1,2,\ldots,r$, the $j$-th row of the check matrix~\eqref{matrixh} reveals the relationship of the bits in a codeword. When we say a bit $s_{\ell,f}$ is repaired by \emph{check-vector} $j$, it means that we access all the bits in the right hand side of \eqref{eq:repair21} when $j=1,2,\ldots,\frac{r}{2}$, \eqref{eq:repair22} when $j=\frac{r}{2}+1,\frac{r}{2}+2,\ldots,r-1$ or \eqref{eq:repair23} when $j=r$ to recover the bit $s_{\ell,f}$.

\begin{algorithm}
\begin{algorithmic}[1]
\STATE {Suppose that the column $f$ has failed.}
\IF {$f\in\{1, 2,\ldots, \lceil n/2 \rceil\}$.}
            \FOR {$\ell \bmod (\frac{r}{2})^f \in\{0,1,\ldots,(\frac{r}{2})^{f-1}-1\}$.}
                  \STATE {Repair the bit $s_{\ell,f}$ by the first check-vector, i.e., by \eqref{eq:repair21} with $j=1$.}
            \ENDFOR
         \FOR {$t=1,2,\ldots,\frac{r}{2}-1$}
            \FOR {$\ell \bmod (\frac{r}{2})^f \in\{t(\frac{r}{2})^{f-1},t(\frac{r}{2})^{f-1}+1,\ldots,(t+1)(\frac{r}{2})^{f-1}-1\}$.}
                  \STATE {Repair the bit $s_{\ell,f}$ by check-vector $\frac{r}{2}-t+1$, i.e., by \eqref{eq:repair21} with $j=\frac{r}{2}-t+1$.}
            \ENDFOR
         \ENDFOR
     \RETURN{}    
    \ENDIF
\IF {$f\in\{\lceil n/2 \rceil+1, \lceil n/2 \rceil+2,\ldots, n\}$.}
            \FOR {$\ell \bmod (\frac{r}{2})^{n+1-f} \in\{0,1,\ldots,(\frac{r}{2})^{n-f}-1\}$.}
                  \STATE {Repair the bit $s_{\ell,f}$ by check-vector $r$, i.e., by \eqref{eq:repair23} with $j=r$.}
            \ENDFOR
\FOR {$t=1,2,\ldots,\frac{r}{2}-1$}
                 \FOR {$\ell \bmod (\frac{r}{2})^{n+1-f} \in\{t(\frac{r}{2})^{n-f},t(\frac{r}{2})^{n-f}+1,\ldots,(t+1)(\frac{r}{2})^{n-f}-1\}$.}
                     \STATE {Repair the bit $s_{\ell,f}$ by check-vector $\frac{r}{2}+t$, i.e., by \eqref{eq:repair22} with $j=\frac{r}{2}+t$.}
            \ENDFOR
         \ENDFOR
    \RETURN{}     
    \ENDIF
    \caption{Repair procedure of the second construction for one column failure.}
    \label{alg:A2}
\end{algorithmic}
\end{algorithm}

The repair algorithm is stated in Algorithm \ref{alg:A2}. In the algorithm, the first $\frac{r}{2}$ check-vectors are used to repair each of the first $\lceil n/2 \rceil$ columns and the last $\frac{r}{2}$ check-vectors are used to repair each of the other $n-\lceil n/2 \rceil$ columns. 

There are some common bits between the bits downloaded by different check-vectors. In Algorithm \ref{alg:A2}, the check-vectors are carefully chosen such that the number of common bits are as large as possible. 
We have shown that the first column of the code in Example \ref{example2} can be repaired by downloading half of the number of bits in each of columns $2,3,4,5,6$. Namely, the repair bandwidth of the first column is weak-optimal.
We show that Algorithm \ref{alg:A2} can recover all the bits of any failed column and the repair bandwidth of one failure column is asymptotically weak-optimal in the next theorem.
\begin{theorem}
Assume that column $f$ is failed. Algorithm \ref{alg:A2} can recover all the bits stored in column $f$, where $1\leq f\leq n$. The repair bandwidth of column $f$ induced by Algorithm \ref{alg:A2} is
\begin{equation}
d(p-1)(\frac{r}{2})^{d-2}+(p-1)((\frac{r}{2})^{d-2}-(\frac{r}{2})^{d-f-1})
\label{eq:repairband2}
\end{equation}
when $1\leq f\leq \lceil n/2 \rceil$, and is
\begin{equation}
d(p-1)(\frac{r}{2})^{d-2}+(p-1)((\frac{r}{2})^{d-2}-(\frac{r}{2})^{d-n+f-2})
\label{eq:repairband3}
\end{equation}
when $\lceil n/2 \rceil+1\leq f\leq n$.
\label{thmrep2}
\end{theorem}
\begin{proof}
Recall that $d=k+\frac{r}{2}-1$ and there are $(p-1)(\frac{r}{2})^{d-1}$ bits in each column. We can recover all $(p-1)(\frac{r}{2})^{d-1}$ bit in column $f$ by a check-vector, i.e., by \eqref{eq:repair21} or \eqref{eq:repair22} or \eqref{eq:repair23}. In Algorithm \ref{alg:A2}, we use $\frac{r}{2}$ check-vectors to recover a column for repair bandwidth reduction. $(p-1)(\frac{r}{2})^{d-1}$ bits in column $f$ are recovered by $\frac{r}{2}$ check-vectors, where each check-vector recovers $(p-1)(\frac{r}{2})^{d-2}$ bits. When $1\leq f\leq \lceil n/2 \rceil$, the number of bits in each column is a multiple of $(\frac{r}{2})^f$ when $k\geq 4$. The values of $\ell \bmod (\frac{r}{2})^f$ are partitioned into $\frac{r}{2}$ groups each with size $(\frac{r}{2})^{f-1}$, where each group is associated with one check-vector. $(p-1)(\frac{r}{2})^{d-2}$ bits to be recovered are then associated with a group according to the values of their indices divided by $(\frac{r}{2})^{f}$. The idea behind this grouping is that we can employ the cyclic structure in the underlying quotient ring to recover some specific bits by choosing a suitable check-vector such that the number of common bits between the bits downloaded by the chosen $\frac{r}{2}$ check-vectors is maximized. When $f\geq \lceil n/2 \rceil+1$, we can replace $\ell \bmod (\frac{r}{2})^{f}$ with $\ell \bmod (\frac{r}{2})^{n+1-f}$ and apply the similar procedure.
In the following, we first show that Algorithm \ref{alg:A2} can recover all the bits in  column $f$ and then present the repair bandwidth of column $f$ by Algorithm \ref{alg:A2}.

Consider $1\leq f\leq \lceil n/2 \rceil$. By steps 3 and 4 in Algorithm \ref{alg:A2}, the bits $s_{\ell,f}$ are recovered by the first check-vector, i.e., by \eqref{eq:repair21} with $j=1$ when
\begin{equation}
\ell\bmod (\frac{r}{2})^{f}\in\{0,1,\ldots,(\frac{r}{2})^{f-1}-1\} 
\label{eq:p21}
\end{equation}
and $\ell\in\{0,1,\ldots,(p-1)(\frac{r}{2})^{d-1}-1\}$. As $(p-1)(\frac{r}{2})^{d-1}$ is a multiple of $(\frac{r}{2})^{f}$ for $k\geq 4$, $\ell\bmod (\frac{r}{2})^{f}$ is uniform distributed over $\{0,1,\ldots,(\frac{r}{2})^{f}-1\}$. So, the number of bits $s_{\ell,f}$ that are recovered by \eqref{eq:repair21} with $j=1$ is $\frac{(p-1)(\frac{r}{2})^{d-1}}{(\frac{r}{2})^{f}}\cdot (\frac{r}{2})^{f-1}=(p-1)(\frac{r}{2})^{d-2}$. 

Recall that there are $(p-1)(\frac{r}{2})^{d-1}$ bits in column $f$ and $(p-1)(\frac{r}{2})^{d-2}$ bits are recovered by~\eqref{eq:repair21} with $j=1$. Hence, we still need to recover the other $(p-1)(\frac{r}{2})^{d-1}-(p-1)(\frac{r}{2})^{d-2}=(\frac{r}{2}-1)(p-1)(\frac{r}{2})^{d-2}$
bits in column $f$. By steps 5 to 7 in Algorithm \ref{alg:A2}, for $t=1,2,\ldots,\frac{r}{2}-1$, the bits $s_{\ell,f}$ are recovered by \eqref{eq:repair21} with $j=\frac{r}{2}-t+1$ when
\begin{equation}
\ell \bmod (\frac{r}{2})^f \in\{t(\frac{r}{2})^{f-1},t(\frac{r}{2})^{f-1}+1,\ldots,(t+1)(\frac{r}{2})^{f-1}-1\}
\label{eq:p22}
\end{equation}
and $\ell\in\{0,1,\ldots,(p-1)(\frac{r}{2})^{d-1}-1\}$. As $\ell\bmod (\frac{r}{2})^{f}$ is uniform distribution, the number of bits repaired by \eqref{eq:repair21} with $j=\frac{r}{2}-t+1$ is $(p-1)(\frac{r}{2})^{d-2}$. Therefore, the number of bits $s_{\ell,f}$ repaired by \eqref{eq:repair21} for $j=2,3,\ldots,\frac{r}{2}$ is $(\frac{r}{2}-1)(p-1)(\frac{r}{2})^{d-2}$.

We have recovered $(p-1)(\frac{r}{2})^{d-2}+(\frac{r}{2}-1)(p-1)(\frac{r}{2})^{d-2}=(p-1)(\frac{r}{2})^{d-1}$ bits $s_{\ell,f}$ with indices in \eqref{eq:p21} and \eqref{eq:p22}. It can be checked that
\begin{align*}
&\{t(\frac{r}{2})^{f-1},t(\frac{r}{2})^{f-1}+1,\ldots,(t+1)(\frac{r}{2})^{f-1}-1\}\cap \\
&\{t'(\frac{r}{2})^{f-1},t'(\frac{r}{2})^{f-1}+1,\ldots,(t'+1)(\frac{r}{2})^{f-1}-1\}=\emptyset
\end{align*}
for $0\leq t\neq t'\leq \frac{r}{2}-1$. Hence, the indices of the repaired bits in column $f$ are distinct and all $(p-1)(\frac{r}{2})^{d-1}$ bits are recovered by Algorithm \ref{alg:A2} for $1\leq f\leq \lceil n/2 \rceil$. Similarly, when $\lceil n/2 \rceil+1\leq f\leq n$, we can also show that we can recover all the bits in column $f$ by Algorithm~\ref{alg:A2}.

Next, we count the repair bandwidth of column $f$ by Algorithm \ref{alg:A2} when $1\leq f\leq \lceil n/2 \rceil$. By steps 3 to 7 in Algorithm \ref{alg:A2}, $(p-1)(\frac{r}{2})^{d-1}$ bits $s_{\ell,f}$ with indices in \eqref{eq:p21} and \eqref{eq:p22} are recovered by \eqref{eq:repair21} with $j=1$ and $j=\frac{r}{2}-t+1$ respectively, where $d(p-1)(\frac{r}{2})^{d-2}$ bits $s_{\ell,1},\ldots,s_{\ell,f-1},s_{\ell,f+1},\ldots,s_{\ell,d+1}$ with indices in \eqref{eq:p21},
$(\frac{r}{2}-1)(d-1)(p-1)(\frac{r}{2})^{d-2}$ bits $s_{\ell+(\frac{r}{2}-t)(\frac{r}{2})^{f-1}-(\frac{r}{2}-t)(\frac{r}{2})^{i-1},i}$ with $i=1,2,\ldots,f-1,f+1,\ldots,d$ and $(\frac{r}{2}-1)(p-1)(\frac{r}{2})^{d-2}$ bits $s_{\ell+(\frac{r}{2}-t)(\frac{r}{2})^{f-1},d+1}$ are needed with $t=1,2,\ldots,\frac{r}{2}-1$. There are many needed bits that are in common so that the total number of downloaded bits is less.

We first download $d(p-1)(\frac{r}{2})^{d-2}$ bits $s_{\ell,1},\ldots,s_{\ell,f-1},s_{\ell,f+1},\ldots,s_{\ell,d+1}$ with indices in \eqref{eq:p21} in steps 3 and 4. Then, we show that we do not need to download all $(\frac{r}{2}-1)d(p-1)(\frac{r}{2})^{d-2}$ bits $s_{\ell+(\frac{r}{2}-t)(\frac{r}{2})^{f-1}-(\frac{r}{2}-t)(\frac{r}{2})^{i-1},i}$ with $i=1,2,\ldots,f-1,f+1,\ldots,d$ and $s_{\ell+(\frac{r}{2}-t)(\frac{r}{2})^{f-1},d+1}$ in steps 5 to 7, as some of them have been downloaded in steps 3 and 4.

We now consider the needed bits $s_{\ell+(\frac{r}{2}-t)(\frac{r}{2})^{f-1}-(\frac{r}{2}-t)(\frac{r}{2})^{i-1},i}$ with $i=1,2,\ldots,f-1,f+1,\ldots,d$ in steps 5 to 7, where $\ell$ are in \eqref{eq:p22}. We want to show that many bits $s_{\ell+(\frac{r}{2}-t)(\frac{r}{2})^{f-1}-(\frac{r}{2}-t)(\frac{r}{2})^{i-1},i}$ are not necessary to be downloaded, as they are downloaded in steps 3 and 4. Given an $\ell$ that is in \eqref{eq:p22}, where $0\leq \ell \leq (p-1)(\frac{r}{2})^{d-1}-1$, we have that the index of the corresponding needed bit $s_{\ell+(\frac{r}{2}-t)(\frac{r}{2})^{f-1}-(\frac{r}{2}-t)(\frac{r}{2})^{i-1},i}$ is $\ell'=(\ell+(\frac{r}{2}-t)(\frac{r}{2})^{f-1}-(\frac{r}{2}-t)(\frac{r}{2})^{i-1})\bmod p(\frac{r}{2})^{d-1}$. Recall that the bits $s_{\ell,i}$ for $i=1,2,\ldots,f-1,f+1,\ldots,d$ and indices in \eqref{eq:p21} are downloaded in steps 3 and 4. If 
\begin{align*}
&\ell'\bmod (\frac{r}{2})^f=\\
&((\ell+(\frac{r}{2}-t)(\frac{r}{2})^{f-1}-(\frac{r}{2}-t)(\frac{r}{2})^{i-1})\bmod p(\frac{r}{2})^{d-1})\bmod (\frac{r}{2})^f \\
&\in \{0,1,\ldots,(\frac{r}{2})^{f-1}-1\},
\end{align*}
then we do not need to download the bit $s_{\ell',i}$, as it is downloaded in steps 3 and 4. Otherwise, it should be downloaded.

We first consider the bits $s_{\ell',i}$ for $i=1,2,\ldots,f-1$. If $\ell\bmod (\frac{r}{2})^f=t(\frac{r}{2})^{f-1}$, then there exists an integer $m$ such that $\ell=m(\frac{r}{2})^f+t(\frac{r}{2})^{f-1}$. Thus, we have
\begin{align*}
&\ell'\bmod (\frac{r}{2})^f\\
=&((\ell+(\frac{r}{2}-t)(\frac{r}{2})^{f-1}-(\frac{r}{2}-t)(\frac{r}{2})^{i-1})\bmod p(\frac{r}{2})^{d-1})\bmod (\frac{r}{2})^f\\
=&(m(\frac{r}{2})^f+t(\frac{r}{2})^{f-1}+(\frac{r}{2}-t)(\frac{r}{2})^{f-1}-(\frac{r}{2}-t)(\frac{r}{2})^{i-1})\bmod (\frac{r}{2})^f \\
&\text{ as } d-1\geq f \text{ for } k\geq 4 \\
=&(\frac{r}{2})^f-(\frac{r}{2}-t)(\frac{r}{2})^{i-1} \text{ as }  f>i \text{ and } t=1,2,\ldots,\frac{r}{2}-1.
\end{align*}
If $\ell\bmod (\frac{r}{2})^f=(t+1)(\frac{r}{2})^{f-1}-1$, then there exists an integer $m$ such that $\ell=m(\frac{r}{2})^f+(t+1)(\frac{r}{2})^{f-1}-1$. Thus, we have
\begin{align*}
&\ell'\bmod (\frac{r}{2})^f\\
=&((\ell+(\frac{r}{2}-t)(\frac{r}{2})^{f-1}-(\frac{r}{2}-t)(\frac{r}{2})^{i-1})\bmod p(\frac{r}{2})^{d-1})\bmod (\frac{r}{2})^f\\
=&(m(\frac{r}{2})^f+(t+1)(\frac{r}{2})^{f-1}-1+(\frac{r}{2}-t)(\frac{r}{2})^{f-1}-\\
&(\frac{r}{2}-t)(\frac{r}{2})^{i-1})\bmod (\frac{r}{2})^f  \\
=&(\frac{r}{2})^{f-1}-(\frac{r}{2}-t)(\frac{r}{2})^{i-1}-1.
\end{align*}
By repeating the above procedure for $\ell\bmod(\frac{r}{2})^f=t(\frac{r}{2})^{f-1}+1,\ldots,(t+1)(\frac{r}{2})^{f-1}-2$, we can obtain that 
\begin{eqnarray}
\ell'\bmod (\frac{r}{2})^f =& (\frac{r}{2})^{f}-(\frac{r}{2}-t)(\frac{r}{2})^{i-1},\ldots,(\frac{r}{2})^{f}-1,0,\nonumber \\
&1,\ldots,(\frac{r}{2})^{f-1}-(\frac{r}{2}-t)(\frac{r}{2})^{i-1}-1
\label{eq:p23}
\end{eqnarray}
when $\ell\bmod (\frac{r}{2})^f$ runs from $t(\frac{r}{2})^{f-1}$ to $(t+1)(\frac{r}{2})^{f-1}-1$. Thus, the bits $s_{\ell',i}$ with $i=1,2,\ldots,f-1$ and $\ell'$ in the union set of all the values in \eqref{eq:p23} with $t=1,2,\ldots,\frac{r}{2}-1$ are needed. 
When $i=1,2,\ldots,f-1$, the union set of all the values in \eqref{eq:p23} with $t=1,2,\ldots,\frac{r}{2}-1$ is
\begin{eqnarray}
&&\cup_{t=1,2,\ldots,\frac{r}{2}-1}\{(\frac{r}{2})^{f}-(\frac{r}{2}-t)(\frac{r}{2})^{i-1},\ldots,(\frac{r}{2})^{f}-1,0,\nonumber \\
&&1,\ldots,(\frac{r}{2})^{f-1}-(\frac{r}{2}-t)(\frac{r}{2})^{i-1}-1\}\nonumber\\
&=&\{(\frac{r}{2})^{f}-(\frac{r}{2}-1)(\frac{r}{2})^{i-1},\ldots,(\frac{r}{2})^{f}-1,0,\nonumber \\
&&1,\ldots,(\frac{r}{2})^{f-1}-(\frac{r}{2})^{i-1}-1\}.
\label{eq:p5}
\end{eqnarray}
Since $i\leq f-1$, we have $(\frac{r}{2})^{f-1}-(\frac{r}{2})^{i-1}-1<(\frac{r}{2})^{f}-(\frac{r}{2}-1)(\frac{r}{2})^{i-1}$. Hence, the elements in~\eqref{eq:p5} can be rearranged as
\begin{align*}
&\{0,1,\ldots,(\frac{r}{2})^{f-1}-(\frac{r}{2})^{i-1}-1,(\frac{r}{2})^{f}-(\frac{r}{2}-1)(\frac{r}{2})^{i-1},\\
&\ldots,(\frac{r}{2})^{f}-1\}.
\end{align*}
Recall that $(f-1)(p-1)(\frac{r}{2})^{d-2}$ bits $s_{\ell,i}$ for $i=1,2,\ldots,f-1$ with $\ell\bmod (\frac{r}{2})^f\in\{0,1,\ldots,(\frac{r}{2})^{f-1}-1\}$ have already been downloaded in steps 3 and 4. Since $(\frac{r}{2})^{f-1}-(\frac{r}{2})^{i-1}-1<(\frac{r}{2})^{f-1}-1<(\frac{r}{2})^{f}-(\frac{r}{2}-1)(\frac{r}{2})^{i-1}$,
\begin{align*}
&\{(\frac{r}{2})^{f}-(\frac{r}{2}-1)(\frac{r}{2})^{i-1},\ldots,(\frac{r}{2})^{f}-1,0,1,\\
&\ldots,(\frac{r}{2})^{f-1}-(\frac{r}{2})^{i-1}-1\}\setminus \{0,1,\ldots,(\frac{r}{2})^{f-1}-1\}\\
=&\{(\frac{r}{2})^{f}-(\frac{r}{2}-1)(\frac{r}{2})^{i-1},\ldots,(\frac{r}{2})^{f}-1\},
\end{align*}
we thus only need to download $(\frac{r}{2}-1)(p-1)(\frac{r}{2})^{d+i-f-2}$ bits $s_{\ell',i}$ for $i=1,2,\ldots,f-1$ with $\ell'\bmod (\frac{r}{2})^f\in\{(\frac{r}{2})^{f}-(\frac{r}{2}-1)(\frac{r}{2})^{i-1},\ldots,(\frac{r}{2})^{f}-1\}$ in steps 5 to 7.

We then consider the bits $s_{\ell',i}$ for $i=f+1,f+2,\ldots,d$. Recall that $\ell'=(\ell+(\frac{r}{2}-t)(\frac{r}{2})^{f-1}-(\frac{r}{2}-t)(\frac{r}{2})^{i-1})\bmod p(\frac{r}{2})^{d-1}$. Similarly, we can prove that 
\[
\ell'\bmod (\frac{r}{2})^f =0,1,\ldots,(\frac{r}{2})^{f-1}-1
\]
when $\ell \bmod (\frac{r}{2})^f$ runs from $t(\frac{r}{2})^{f-1}$ to $(t+1)(\frac{r}{2})^{f-1}-1$. As all needed bits in this case have already been downloaded in steps 3 and 4, 
we thus do not need to download bits from columns $i$ for $i=f+1,f+2,\ldots,d$ in steps 5 to 7.

In the last, we consider $i=d+1$, i.e., the needed bits $s_{\ell+(\frac{r}{2}-t)(\frac{r}{2})^{f-1},d+1}$ with indices in~\eqref{eq:p22}.
Given an $\ell$ that is in \eqref{eq:p22}, where $0\leq \ell \leq (p-1)(\frac{r}{2})^{d-1}-1$, we have that the index of the corresponding needed bit $s_{\ell+(\frac{r}{2}-t)(\frac{r}{2})^{f-1},d+1}$ is $\ell'=(\ell+(\frac{r}{2}-t)(\frac{r}{2})^{f-1})\bmod p(\frac{r}{2})^{d-1}$.
If $\ell\bmod (\frac{r}{2})^f=t(\frac{r}{2})^{f-1}$, then $\ell=m(\frac{r}{2})^f+t(\frac{r}{2})^{f-1}$, where $m$ is an integer. Thus, we have
\begin{align*}
&\ell'\bmod (\frac{r}{2})^f\\
=&((\ell+(\frac{r}{2}-t)(\frac{r}{2})^{f-1})\bmod p(\frac{r}{2})^{d-1})\bmod (\frac{r}{2})^f\\
=&(m(\frac{r}{2})^f+t(\frac{r}{2})^{f-1}+(\frac{r}{2}-t)(\frac{r}{2})^{f-1})\bmod (\frac{r}{2})^f \\
=&0.
\end{align*}
Hence, we can obtain that 
\[
\ell'\bmod (\frac{r}{2})^f =0,1,\ldots,(\frac{r}{2})^{f-1}-1
\]
when $\ell \bmod (\frac{r}{2})^f$ runs from $t(\frac{r}{2})^{f-1}$ to $(t+1)(\frac{r}{2})^{f-1}-1$. Recall that $(p-1)(\frac{r}{2})^{d-2}$ bits $s_{\ell,d+1}$ with indices in \eqref{eq:p21} have already been downloaded in steps 3 and 4, 
we thus do not need to download bits $s_{\ell,d+1}$ in steps 5 to 7.

Therefore, the total number of bits downloaded from $d=k+\frac{r}{2}-1$ columns to repair column $f$ is 
\begin{align*}
&d(p-1)(\frac{r}{2})^{d-2}+((\frac{r}{2})-1)(p-1)\sum_{i=1}^{f-1}(\frac{r}{2})^{d-f+i-2}\\
=&d(p-1)(\frac{r}{2})^{d-2}+(p-1)((\frac{r}{2})^{d-2}-(\frac{r}{2})^{d-f-1}),
\end{align*}
which is \eqref{eq:repairband2}. 

When $f> \lceil n/2 \rceil$, we can replace $\ell \bmod (\frac{r}{2})^f$ with $\ell \bmod (\frac{r}{2})^{n+1-f}$ and obtain the repair bandwidth of column $f$ that is \eqref{eq:repairband3} with the same argument. This completes the proof.
\end{proof}

Note that the repair bandwidth of column $n+1-f$ is the same as the repair bandwidth of column $f$ for $f=1,2,\ldots,\lceil n/2\rceil$, according to \eqref{eq:repairband2} and \eqref{eq:repairband3}. When $f=1$, the repair bandwidth is $d(p-1)(\frac{r}{2})^{d-2}$, which is weak-optimal. Consider the worst case of $f=\lceil n/2\rceil$, the repair bandwidth is
\begin{align*}
&(p-1)((d+1)(\frac{r}{2})^{d-2}-(\frac{r}{2})^{d-\lceil n/2\rceil-1})\\
&< (p-1)(d+1)(\frac{r}{2})^{d-2},
\end{align*}
which is strictly less than $\frac{d+1}{d}$ times of the value in \eqref{optimal_repair}. The repair bandwidth of any one failure can achieve the weak-optimal repair in \eqref{optimal_repair} asymptotically when $d$ is large enough.

\section{Conclusion}

\label{sec:discussions}
We present an approach of designing binary MDS array code over a specific binary cyclic ring. Two constructions based on the coding approach by choosing the well-designed encoding matrix and parity matrix respectively are proposed. Two efficient repair methods for any one information failure of the first construction and for any one failure of the second construction are designed. We also show that the repair bandwidth of these two repair methods is asymptotically weak-optimal. 



Note that the repair bandwidth of the first construction is asymptotically weak-optimal only for a single information failure instead of for a parity failure. A generic transformation given in~\cite{tian2017generic} that can be applied on non-binary MDS codes to produce new MDS codes such that some nodes of the new MDS codes have weak-optimal repair bandwidth. A future work is to modify the transformation method in \cite{tian2017generic} and to employ it to enable weak-optimal repair bandwidth for parity column, while the asymptotically weak-optimal repair bandwidth of information column is still maintained.

In some practical applications, the system might want to repair more than one information columns simultaneously. Another interesting future work is the design of efficient repair algorithm when more than one information columns fail. 


\appendices

\bibliographystyle{IEEEtran}


\end{document}